\documentclass[11pt]{article}
\usepackage{fullpage}
\usepackage{amsmath}
\usepackage{amsthm}
\usepackage{amsfonts}
\usepackage{amssymb}
\usepackage{graphics}
\usepackage{graphicx}
\usepackage{amsbsy}
\usepackage{verbatim}
\usepackage{enumerate}
\usepackage{color}
\usepackage{epstopdf}
\usepackage{placeins,sidecap}
\newtheorem{remark}{Remark}[section]

\newtheorem{prop}{Proposition}[section]
\newtheorem{lemma}{Lemma}[section]
\newcommand{\bX}{{\mathbb X}}

\newcommand{\cK}{{\mathcal K}}
\newcommand{\cH}{{\mathcal H}}
\newcommand{\Id}{\mathsf{\mathrm{I}}}
\newcommand{\ds}{\ensuremath{\displaystyle}}
\newcommand{\dy}{\ensuremath{\mathrm{d} y}}
\newcommand{\de}{\ensuremath{\mathrm{d}}}
\newcommand{\T}{{\mathbb T}}
\newcommand{\R}{{\mathbb R  }}

\newcommand{\cA}{\mathcal A}
\newcommand{\cS}{\mathcal S}
\newcommand{\cL}{\mathcal L}
\newcommand{\hcL}{\hat{\cL}}

\newcommand{\eps}{\epsilon}

\newcommand{\EXP}[1]{\ensuremath{\mathrm{e}^{#1}}}
\newcommand{\minus}{-}
\newcommand{\plus}{+}
\newcommand{\Dmat}{\mathsf{\mathrm{D}}}
\newcommand{\Dmato}{\Dmat_0}
\newcommand{\tp}{t^{\prime}}
\newcommand{\Ran}{\,\mathsf{\mathrm{Ran}}\,}

\begin{document}
\setlength{\baselineskip}{15pt}
%\title{Brownian Motion in Tilted Periodic Potentials}
\title{Corrections to Einstein's relation for Brownian motion in a tilted periodic potential}
\author{J.C. Latorre \\
        Department of Mathematics and Computer Science,\\ Freie Universit\"{a}t Berlin, Germany \\
      G.A. Pavliotis\footnote{Corresponding author.
E-mail address: g.pavliotis@imperial.ac.uk} \\
        Department of Mathematics\\
    Imperial College London \\
        London SW7 2AZ, UK \\
        and  \\
         Department of Mathematics and Computer Science,\\ Freie Universit\"{a}t Berlin, Germany \\
       P.R. Kramer\\
 Mathematical Sciences Department,\\ Rensselaer Polytechnic Institute, \\ 
        110 8th Street, Troy, New York 12180, USA
                    }
           
\maketitle

\begin{abstract}
In this paper we revisit the problem of Brownian motion in a tilted periodic potential. We use homogenization theory to derive general formulas for the effective velocity and the effective diffusion tensor that are valid for arbitrary tilts. Furthermore, we obtain power series expansions for the velocity and the diffusion coefficient as functions of the external forcing. Thus, we provide systematic corrections to Einstein's formula and to linear response theory. Our theoretical results are supported by extensive numerical simulations. For our numerical experiments we use a novel spectral numerical method that leads to a very efficient and accurate calculation of the effective velocity and the effective diffusion tensor.
\end{abstract}
%%%%%%%%%%%%%%%%%%%%%%%%%%%%%%%%%%%%%%%%%%%%%%%%%%%%%%%%%%%%%%%%%%%%%%%%%%%%%%
%
%                                      INTRODUCTION
%
%%%%%%%%%%%%%%%%%%%%%%%%%%%%%%%%%%%%%%%%%%%%%%%%%%%%%%%%%%%%%%%%%%%%%%%%%%%%%%
\section{Introduction}
\label{sec:intro}

Brownian motion in periodic potentials is one of standard models in condensed matter physics. Applications include the modeling of Josephson junctions, polymer dynamics, superionic conduction, dielectric relaxation, plasma physics and surface diffusion~\cite{Ferrando_al02}. A detailed discussion and extensive bibliography can be found in~\cite{coffey04, Ris84}. 

%\textbf{PRK:  Rewrote to make it clearly multidimensional.}
The goal of this paper is to study Brownian motion in a tilted periodic potential for arbitrary values of the drift and of the tilt (external forcing). The dynamics of the Brownian particle is governed by the Langevin equation
\begin{equation}\label{e:langevin_intro}
\ddot{\boldsymbol{q}} = - \boldsymbol{\nabla}_{\boldsymbol{q}} V(\boldsymbol{q}) +\boldsymbol{F} - \gamma \dot{\boldsymbol{q}} + \sqrt{2 \gamma \beta^{-1}} \dot{\boldsymbol{W}},
\end{equation}
where $V$ is a periodic potential with period $L$ (in each direction), $\boldsymbol{F}$ denotes a constant external force, so that the effective potential is 
\begin{equation}
V_{\mbox{eff}}(\boldsymbol{q}) = V(\boldsymbol{q}) - \boldsymbol{F}\cdot \boldsymbol{q}, \label{e:effpot}
\end{equation}
 $\gamma$ is the friction coefficient, $\beta$ is the inverse temperature
and $\boldsymbol{W}(t)$ is a standard Brownian motion on $\R^d$. 

%\textbf{PRK:  Vector notation introduced again}
The main objective is the calculation of the drift and diffusion coefficients which are defined as
\begin{equation}
\boldsymbol{U} = \lim_{t \rightarrow \infty} \frac{\langle \boldsymbol{q}(t) - \boldsymbol{q}(0) \rangle}{t}
\label{e:drift}
\end{equation}
and
\begin{equation}
\boldsymbol{\Dmat} = \lim_{ t \rightarrow \infty} \frac{\langle \boldsymbol{q}(t) - \langle \boldsymbol{q}(t)
 \rangle ) \otimes ( \boldsymbol{q}(t) - \langle \boldsymbol{q}(t)  \rangle ) \rangle}{2 t}.
\label{e:diff}
\end{equation}
Here $\langle \cdot \rangle$ denotes ensemble average and $\otimes$ stands for the tensor product. Explicit formulas for these coefficients are available only in  
 the overdamped limit and mostly in one dimension.  An exact analytical formula for the effective velocity of an overdamped Brownian particle moving in a one dimensional tilted periodic potential was obtained many years ago by Stratonovich~(\cite{Straton58},\cite[Ch. 9]{Straton67}). A corresponding analytical formula for the diffusion coefficient was obtained and analyzed more recently~\cite{reimann_al01, reimann_al02, pavl05}, and verified in an experimental realization of the model involving rotating optical tweezers~\cite{EvstigReimann2008}.  Simpler algebraic formulas were deduced from these for the special case of piecewise linear potentials in~\cite{Heinsalu2004}.
Only potentials with very specific geometries can lead to analytical formulas in dimension higher than one~\cite{ferrando_all92, Caratti_all96}. A wealth of information on the problem of Brownian motion in a tilted periodic potential in one dimension can be found in~\cite[Ch. 11]{Ris84}.

It is well known that the equilibrium diffusion coefficient (i.e., the diffusion coefficient in the absence of an external drift) and the drift, or, rather, the mobility are related through the famous Einstein formula:
\begin{equation}\label{e:einstein}
\boldsymbol{\Dmato} = \beta^{-1} \boldsymbol{\mu},
\end{equation}
with %(\textbf{PRK:  Please check this is OK in multi-D; we should write like this if we're going to frame our work as applicable to multi-d})
$$
\boldsymbol{\mu} = \lim_{|\boldsymbol{F}| \rightarrow 0} \boldsymbol{\nabla}_{\boldsymbol{F}} \boldsymbol{U}
$$
The validity of this formula has been proved rigorously for several models~\cite{lebo_einstein}, including that of a Brownian particle in a tilted periodic potential~\cite{rodenh}. Formulas of the form~\eqref{e:einstein} can be understood in the more general framework of linear response theory and of the Green-Kubo formalism~\cite{ResibDeLeen77, KuboTodaHashitsume91}. A recent rigorous analysis of the Green-Kubo formalism for the calculation of the shear viscosity can be found in~\cite{JoubaudStoltz2011}.

The main goal of the present paper is to investigate the validity and usefulness of corrections to linear response theory. In particular, we calculate all terms in the power series expansions (with respect to the forcing $F$) for the drift and diffusion coefficients and we use these in order to calculate corrections to Einstein's formula~\eqref{e:einstein}. Our analysis is based on the formalism of averaging and homogenization~\cite{PavlSt08}. From this formalism we know that both drift and diffusion coefficients can be expressed in terms of the solution of appropriate Poisson equations \eqref{e:fp_station}-\eqref{e:cell}.
% \begin{subequations}\label{e:fp_poisson_intro}
% \begin{eqnarray}\label{e:fp_intro}
% \cL^{*} \rho_{\beta} &=& 0, \quad U = \int_{\R^D}\int_{[0,L]^d} p \rho_{\beta} \, dp dq, \quad \int_{\R^d}\int_{[0,L]^d} \rho_{\beta} \, dp dq = 1,
% \\
% - \cL \phi &=& p - U, \quad D = \int_{\R^d}\int_{[0,L]^d} (p-U) \phi \rho_{\beta} \, dp dq, \quad \int_{\R^d}\int_{[0,L]^d} \phi \rho_{\beta} \, dp dq= 0, 
% \label{e:poisson_intro}
% \end{eqnarray} 
% \end{subequations}
% where $\cL$ and $\cL^{*}$ denote the generator and the Fokker-Planck operator of the Markov process $\{q, \, p \}$ with $p = \dot{q}$. Equations~\eqref{e:fp_poisson_intro} are equipped with appropriate boundary conditions. The centering condition in~\eqref{e:poisson_intro} ensures the uniqueness of $\phi$. 
Details are presented in the next section. 

For simplicity of notation and presentation, we will restrict our calculations for corrections to linear response theory to the one dimensional case $ d=1 $, and hence hereafter drop vector and tensor notation.  Completely analogous formulas are applicable in multiple dimensions. 
We present our results in detail in Section~\ref{sec:lin_resp}, and summarize them here rather imprecisely:

\begin{prop}\label{prop:main}
The drift and diffusion coefficients admit the asymptotic expansions 
\begin{equation}\label{e:expansion}
U =  \sum_{\ell \geq 1} F^{\ell} V_{\ell} \quad ,
\end{equation}
and
\begin{subequations}
\begin{eqnarray}\label{e:deff_coeffs}
D &=&  \sum_{\ell \geq 0} F^{\ell} \left[\beta^{-1}V_{\ell+1} +\sum_{n=1}^{\ell} \Sigma_{n\ell}\right] \\
&=& \beta^{-1} \frac{d U}{dF} + \sum_{\ell \geq 1} F^{\ell} \sum_{n=1}^{\ell} \Xi_{n \ell}.
%\label{e:deff_coeffs2}
\end{eqnarray}
\end{subequations}

% \begin{equation}\label{e:deff_coeffs}
% D = \beta^{-1} \sum_{\ell \geq 0} F^{\ell} V_{\ell+1} + \sum_{r =1 }^{n} \sum_{\ell \geq r} F^{\ell} \Sigma_{n \ell}.
% \end{equation}
The coefficients $ V_{j}, \, \Sigma_{nj}, \, \Xi_{nj}; n, \, j=1,2,\dots$  can be computed in terms of solutions to Poisson equations for the generator of the equilibrium dynamics $F=0$.  In particular, the higher order corrections to the drift and diffusion coefficients are not compatible with an extension of the Einstein relation (\ref{e:einstein}) beyond the linear response regime $ F \rightarrow 0 $. 

For the case of a symmetric potential ($ V(q) =V(-q) $), then $ V_n = 0 $ for even $ n $ and $ \Sigma_{n \ell} = \Xi_{n \ell} = 0 $  for odd $ \ell $. 
\end{prop}
Thus, it is possible, in principle, to calculate the drift and diffusion coefficients of the nonequilibrium dynamics~\eqref{e:langevin_intro} in terms of the equilibrium dynamics
\begin{equation}\label{e:langevin_eq}
\ddot{q} = - V^{\prime}(q) - \gamma \dot{q} + \sqrt{2\gamma \beta^{-1}} \dot{W}
\end{equation}
for at least a finite interval of values of $F$.

The validity and usefulness of the power series expansions~\eqref{e:expansion} is tested by performing numerical experiments. For the calculation of drift and diffusion coefficients we need to solve Poisson equations of the form
\begin{equation}\label{e:poisson_introII}
- \cL \phi = f(p,q),
\end{equation}
with $\cL$ being the generator of the Markov process $\{q, \, p \}$ with $p = \dot{q}$. We solve equations of the form~\eqref{e:poisson_introII} using a spectral method~\cite{PavlVog08} that is an extension of Risken's continued fraction expansion method~\cite{Ris84}. By comparing the results obtained using our spectral method with results obtained from (the computationally more expensive) Monte Carlo simulations, we demonstrate that our method performs very well.

The rest of the paper is organized as follows. In Section~\ref{sec:lang} we present the formulas for the drift and diffusion coefficients obtained using homogenization theory. In Section~\ref{sec:lin_resp} we calculate the power series expansions for the drift and the diffusion coefficient. In Section~\ref{sec:numerics} we present results of numerical simulations on the calculation of $U$ and $D$. Section~\ref{sec:conclusions} summarizes our conclusions. The details of the spectral method for the solution to the Poisson equation are presented in Appendix~\ref{sec:num_alg}.  Some discussion of how our formulas relate to an alternative approach developed in~\cite{BaiesiMaes11} can be found in Appendix~\ref{sec:appmaes}. 

%
%%%%%%%%%%%%%%%%%%%%%%%%%%%%%%%%%%%%%%%%%%%%%%%%%%%%%%%%%%%
%
%             LANGEVIN EQUATION
%
%%%%%%%%%%%%%%%%%%%%%%%%%%%%%%%%%%%%%%%%%%%%%%%%%%%%%%%%%%%
%
\section{Formulas for the Drift and Diffusion Coefficients}
\label{sec:lang}

We start by writing~\eqref{e:langevin_intro} as a first order system, in $ d=1 $ dimension:
\begin{equation}
\dot{q} = p, \quad \dot{p} = -  V^{\prime}(q) +F - \gamma p + \sqrt{2 \gamma \beta^{-1}} \dot{W}.  \label{e:sdepq}
\end{equation}
The process $\{q, \, p \}$ is a Markov process with generator
\begin{equation}\label{e:generator}
\cL = p \cdot \partial_q + (-\partial_q V + F) \cdot \partial_p + \gamma \big(
- p \cdot \partial_p + \beta^{-1} \partial^{2}_p \big).
\end{equation}
The Fokker-Planck operator, i.e. the $L^{2}$--adjoint of the generator, is
\begin{equation}\label{e:fokker_planck}
\cL^{*} = -p \cdot \partial_q + (\partial_q V - F) \cdot \partial_p + \gamma \partial_{p} \cdot \big(p + \beta^{-1} \partial_p \big).
\end{equation}
The potential function $ V $ has period $L $.  
We can use homogenization theory~\cite{pavl05, HairPavl04, PavSt05b} to prove that the rescaled process
\begin{equation}\label{e:resc}
q^{\eps}(t):= \eps q(t/\eps^{2}) - \frac{t U}{\eps},
\end{equation}
where $U$ is the effective drift as defined below, converges weakly on $C([0,T] ; \R)$ to a Brownian motion with diffusion coefficient $D$. To write down the formulas for the drift $U$ and the diffusion coefficient $D$ we need to consider the process $\boldsymbol{X} (t) = (q(t), \,p(t))$ defined on $\bX :=  \T \times \R$ where $\T$ denotes a one-dimensional circle with length $L $ corresponding to the period of the potential $ V$.  The generator and Fokker-Planck operator of this process are still given by formulas~\eqref{e:generator} and~\eqref{e:fokker_planck} but now restricted on $\bX$, with periodic boundary conditions with respect to $q$. It can be shown~\cite{rodenh} that $\boldsymbol{X}(t)$ is an ergodic Markov process with invariant measure $\mu_{\beta}(dp dq)$ that has a smooth density $\rho_{\beta}(p,q)$ with respect to the Lebesgue measure on $\bX$. The invariant density is the unique solution of the stationary Fokker-Planck equation on $\bX$:
\begin{equation}\label{e:fp_station}
\cL^{*} \rho_{\beta} = 0.
\end{equation}
The drift is then given by the average of the momentum with respect to $\rho_{\beta}$ over $\bX$:
\begin{equation}\label{e:drift_defn}
U = \int_{\bX} p \rho_{\beta}(p,q) \, dp dq.
\end{equation}

The diffusion coefficient is given by the formula
\begin{subequations}
\begin{eqnarray}
\label{e:deff}
D  & = & \int_{\bX} (p - U) \, \phi \rho_{\beta}(p,q) \, dp dq 
\\ & = &
\gamma \beta^{-1} \int_{\bX} \left( \partial_{p} \phi \right)^2 \,   \rho_{\beta} (p,q) \, dp dq , \label{e:deff1}
\end{eqnarray}
\end{subequations}
where $\phi$ is the solution of the Poisson equation 
\begin{equation}\label{e:cell}
- \cL \phi = p-U, \quad \int_{\bX} \phi \rho_{\beta}  \, dp \, dq= 0.
\end{equation}
Equations~\eqref{e:fp_station} and~\eqref{e:cell} are equipped with periodic boundary conditions in $q$ and suitable integrability conditions~\cite{pavl05, HairPavl04, PavSt05b}. Formula~\eqref{e:deff1}, which shows that the effective diffusion tensor is positive semidefinite, follows from~\eqref{e:deff} after an integration by parts. 

It is possible to prove that both $U$ and $D$ are analytic functions of the forcing $F$. This has been proved for the drift in~\cite{ColletMatinez2008} (in fact, in this paper the analyticity of the drift with respect to the forcing is proved for several models including systems of coupled Fokker--Planck equations). A similar analysis can be used to prove the analytic dependence of $D$ on $F$.
%
%
%%%%%%%%%%%%%%%%%%%%%%%%%%%%%%%%%%%%%%%%%%%%%%%%%%%%%%%%%%%
%        LINEAR RESPONSE
%
%
%%%%%%%%%%%%%%%%%%%%%%%%%%%%%%%%%%%%%%%%%%%%%%%%%%%%%%%%%%%
%
\section{Corrections to Linear Response Theory}
\label{sec:lin_resp}
In this section we solve perturbatively equations~\eqref{e:fp_station} and~\eqref{e:cell} in one dimension, in order to obtain the power series expansions~\eqref{e:expansion}. Calculations of this form are quite standard when investigating the effect of colored noise on the drift and diffusion coefficients, e.g.~\cite{HorsLef84, DoerDonKlos98, pavl05}. Related calculations have presented recently in~\cite{JoubaudStoltz2011}.

The main result of this section is a precise formulation of Proposition~\ref{prop:main}. To state the result, we need to introduce some notation. We denote by $H_{0}$ the Hamiltonian of the unperturbed (equilibrium) dynamics~\eqref{e:langevin_intro}:
$$
H_{0}(p,q) = \frac{1}{2} p^{2} + V(q).
$$
The invariant density of the unperturbed dynamics on $\bX$ is denoted by $\bar{\rho}$:
\begin{equation}
\bar{\rho} (q,p) = \frac{1}{Z} e^{-\beta H_{0}(q,p)}, \quad Z = \int_{\bX} e^{- \beta H_{0}(q,p)} \, dq dp.
\label{e:undens}
\end{equation}
We will work in the weighted $L^{2}$ space $\cH =L^{2}(\T \times \R ; \bar{\rho})$. The inner product in this Hilbert space will be denoted by $\langle \cdot, \cdot \rangle_{\beta}$. The generator of the unperturbed dynamics can be written in the form
\begin{equation}\label{e:L0}
\cL_{0} = \cA + \gamma \cS,
\end{equation}
where
$$
\cA = p \partial_q - \partial_q V \partial_p \quad \mbox{and} \quad \cS = -p \partial_p + \beta^{-1} \partial_p^2,
$$
denote the reversible and irreversible parts, respectively. The operators $\cA$ and $\cS$ are antisymmetric and symmetric, respectively, in $\cH$. We introduce now the creation and annihilation operators~\cite{Tit78,Ris84,GlimmJaffe87,HelNi05} 
\begin{equation}\label{e:cre_annh}
a^+ := -\partial_p + \beta p \quad  \mbox{and} \quad a^- :=\partial_p.
\end{equation}
These two operators are $\cH$-adjoint:
$$
\langle a^+ f, h \rangle_\beta = \langle f,a^{-} h \rangle_\beta, \quad \forall \; f, \, h \in \cH.
$$
\begin{prop}\label{prop:detailed}
The drift and diffusion coefficients admit the asymptotic expansions
\begin{equation}\label{e:expansion1}
U =  \sum_{\ell \geq 1} F^{\ell} V_{\ell}  
\end{equation}
and 
\begin{subequations}
\begin{eqnarray}\label{e:deff_coeffs1}
D &=&  \sum_{\ell \geq 0} F^{\ell} \left[\beta^{-1}V_{\ell+1} +\sum_{n=1}^{\ell} \Sigma_{n\ell}\right] \\
&=& \beta^{-1} \frac{d U}{dF} + \sum_{\ell \geq 1} F^{\ell} \sum_{n=1}^{\ell} \Xi_{n \ell}.
\label{e:deff_coeffs2}
\end{eqnarray}
\end{subequations}
The coefficients $ V_{\ell}, \, \Sigma_{k \ell}, \, \Xi_{k \ell} \; k \leq \ell=1,2,\dots$ are given by the formulas 
%(\textbf{JCL: We fixed the indices and integrals in the definitions of Sigma and Xi.  PRK:  Well you can do that, but then the summations should be from $0 $ to $ \ell-1 $, which makes listing indices more awkward because $ \ell $ starts at 1.  If you prefer this choice, that's fine, but the number of terms in the sum for $ \Sigma $ is $ \ell $; it's either from $ 0 $ to $ \ell-1$ (your notation) or $ 1 $ to $ \ell $ (my notation).  Or tell me what I'm seeing wrong; more comment in derivation below.})
\begin{subequations}
\begin{eqnarray}
\label{e:vel_j}
V_\ell &=& \int_{\bX}  f_{\ell} p \bar{\rho} \, dp dq = \beta \int \phi_{\ell-1} p \bar{\rho} \, dp dq,
\\
\label{e:deff_nj}
\Sigma_{n \ell} &=& \int_{\bX} p \phi_{\ell-n} f_{n}  \bar{\rho} \, dp dq
\\
\label{e:deff_nj2}
\Xi_{n \ell} &=& \beta^{-1} \int_{\bX}  \phi_{\ell-n} \partial_p f_{ n} \bar{\rho} \, dp dq
\end{eqnarray}
\end{subequations}
where $f_{j}, \, \phi_{j}, \; j = 0, \dots$ are solutions to the (adjoint) Poisson equations
\begin{subequations}\label{e:ll*}
\begin{eqnarray}\label{e:l*}
- \hat\cL_0 f_j &=& a^+ f_{j-1}, \quad f_0 =1, \quad \int f_j \bar{\rho} =0, \;\; \quad j=1, 2, \dots,
\\
-\cL_0 \phi_0 &=& p, \quad \int \phi_0 \bar{\rho} =0, \\
\label{e:l}
- \cL_0 \phi_j &=& a^-\phi_{j-1} - V_j,   \quad
 \int \phi_j \bar{\rho} = - \sum_{r=1}^j \int f_r \phi_{j-r} \bar{\rho}
 \;\;  j=1, \dots
\end{eqnarray}
\end{subequations}
We have used the notation $\hat\cL_0 = -\cA + \gamma \cS$ to denote the $\cH$-adjoint of the generator $\cL_{0}$.
\end{prop}

\begin{remark}
The expansion formulas for the drift and velocity are consistent with the exact statistical reflection symmetry $ q \rightarrow - q $, $ p \rightarrow -p $, and $ F \rightarrow - F $ in the stochastic system~\eqref{e:sdepq} or infinitesmal generator~\eqref{e:generator} when the potential is symmetric:  $ V(q) = V(-q) $.  Since the drift is odd and the diffusivity even under reflection, this implies that the coefficients $ V_{\ell} = 0 $ when $ \ell $ is even and the coefficients $ \Sigma_{n \ell} = \Xi_{n \ell} = 0 $ when $ \ell $ is odd.  One can verify that our formulas do indeed have these vanishing properties under symmetry of the potential, noting that the operators $ a^+ $ and $ a^- $ are odd under reflection, whereas $ \cL_0 $ and $ \hat\cL_0 $ are even.  By uniqueness of solutions of the Poisson equations~\eqref{e:ll*} and the symmetry properties of the operators, inhomogeneity, and auxiliary conditions, we first verify by induction that the functions $ f_j $ have even reflection symmetry for even $ j $ and odd reflection symmetry for odd $ j $.  Then we similarly induce that the functions $ \phi_j $ have odd reflection symmetry for even $j $ and even reflection symmetry for odd $ j $.  Finally, $ \bar{\rho} $~\eqref{e:undens} has manifestly even symmetry under reflection symmetry.  Therefore, when $ \ell $ is even, $ V_{\ell} $ can be checked to be the periodic integral of an odd function and when $ \ell $ is odd, $ \Xi_{n \ell} $ and $ \Sigma_{n \ell} $ are periodic integrals of odd functions, and so vanish. 
\end{remark}

Using the notation that we have introduced in this section, Einstein's formula (linear response theory) can be written in the form
$$
D_{0}= \beta^{-1} V_1. 
$$
However, formula~\eqref{e:deff_coeffs1} shows that it is not true that a similar simple relation holds for higher order terms in the expansions for the drift and the diffusion coefficients. In particular, {\bf it is not true} that
\begin{equation} \label{e:d_lin_resp_assumption}
D_{n}=\beta^{-1}  V_{n+1}, \quad n=1, \dots, 
\end{equation}
but instead there is a non-trivial correction to~\eqref{e:d_lin_resp_assumption} that is given by the second term on the right hand side of~\eqref{e:deff_coeffs1}.  As an example, we present the formula for the diffusion coefficient that is valid up to $\mathcal{O}(F^{3})$:
\begin{eqnarray}\label{e:deff_correct}
D &=& \beta^{-1} V_{1} + F \Big( \beta^{-1} V_{2} + \int_{\bX} p \phi_{0} f_{1} \bar{\rho}  \, dp dq \Big)
\\ & + & \qquad 
       F^{2} \Big( \beta^{-1} V_{3} + \int_{\bX} p \phi_{1} f_{1} \bar{\rho} \, dp dq + \int_{\bX} p \phi_{0} f_{2} \bar{\rho} \, dp dq \Big) + \mathcal{O}(F^{3}). \nonumber
\end{eqnarray}
Notice that the calculation of the next two terms in the expansion for the diffusion coefficient requires the solution of an additional Poisson equation, in order to compute $\phi_{1}$, as well as the calculation of three additional integrals.

Similarly, \textbf{it is not true} that the Einstein relation~\eqref{e:einstein} can be extended away from $ F = 0 $ in the form:
\begin{equation}
D (F) = \beta^{-1} \frac{d U (F)}{d F}, \label{e:exeinstein}
\end{equation}
because of the presence of correction terms in Eq.~\eqref{e:deff_coeffs2}. This issue is investigated numerically in the next section, see Figures~\ref{u_linear_response} and~\ref{d_linear_response}.  The relation Eq.~\eqref{e:exeinstein} was indeed hypothesized in~\cite{CostantiniMarchesoni1999}, but~\cite{reimann_al02} showed through analytical and numerical studies that while it seems qualitatively correct, and is quantitatively correct in the three  limits $ F \rightarrow 0 $, $ F \rightarrow \infty $, and $ \beta \rightarrow \infty $, it is not quantitatively accurate for general parameter values.  Our results in Proposition~\ref{prop:detailed} give quantitative formulas for this discrepancy, for example, through third order:
\begin{eqnarray}
D &=& \beta^{-1} \frac{d U(F)}{d F} + F \beta^{-1} \int_{\bX} \phi_0 \partial_p f_1 \bar{\rho} \, dp dq 
\label{eq:EinstXiex}\\
& + & \qquad F^2 \beta^{-1}  \Big(\int_{\bX} \phi_0 \partial_p f_2 \bar{\rho} \, dp dq
+ \int_{\bX} \phi_1 \partial_p f_1 \bar{\rho} \, dp dq \Big).
\end{eqnarray}

The violation of the Einstein relation for $ F \neq 0 $ in the model under consideration, and other nonequilibrium steady-state models, was recently analyzed by~\cite{BaiesiMaes11} from a different nonperturbative perspective, expressing the correction terms with respect to various time-correlation functions of the dynamics.  But as we discuss in Appendix~\ref{sec:appmaes}, our framework based on perturbation expansions of the equations from homogenization theory appear to yield more easily computable expressions. We remark also that~\cite{KomOlla2005} have examined deviations from the Einstein relation in the context of stochastic tracer dynamics in a random environment. %Their main result,~\cite[Thm. ]{KomOlla2005}

{\it Proof of Prop.~\ref{prop:detailed}.} We start with the analysis of the stationary Fokker-Planck equation~\eqref{e:fp_station}. We set
\begin{equation}\label{e:f}
\rho_{\beta}(p,q) = \bar{\rho} (p,q) f(p,q)
\end{equation}
We substitute~\eqref{e:f} into~\eqref{e:fp_station} and use the symmetry and antisymmetry of $\cS$ and $\cA$, respectively as well as equation~\eqref{e:cre_annh} to obtain
\begin{equation}\label{e:f_eqn}
\Big(\hat\cL_{0} + F a^+ \Big) f = 0.
\end{equation}
Equation~\eqref{e:f_eqn} is posed on $\bX := \T \times \R$ and is equipped with periodic boundary conditions with respect to $q$ as well as the normalization condition
$$
\int_{\bX} f \bar{\rho} \, dp dq =1.
$$
We look for a solution to~\eqref{e:f_eqn} as a power series expansion in $F$:
\begin{equation}
f(p,q) = \sum_{j=0}^N F^j f_j (p,q). \label{e:fexp}
\end{equation}
The normalization condition becomes
$$
\sum_{j=0}^N F^j \int_{\bX} f_j (p,q) \bar{\rho} (p,q) \, dp dq = 1.
$$
This condition has to be satisfied for all $F \in \R$ which implies that the following normalization conditions should be satisfied
\begin{equation}\label{e:normaliz}
\langle f_0 , {\bf 1}\rangle_\beta =1, \quad \langle f_j , {\bf 1}\rangle_\beta = 0, \;\; j=1, 2, \dots
\end{equation}
We substitute the expansion for $f$ into~\eqref{e:f_eqn} to obtain the sequence of equation
\begin{subequations}\label{e:f_j_eqns}
\begin{eqnarray}
\hat\cL_0 f_0 &=& 0,
\\
\label{e:f_j}
-\hat\cL_0 f_j &=& a^+ f_{j-1}, \quad j=1,2, \dots
\end{eqnarray}
\end{subequations}
The above equations are of the form
\begin{equation}\label{e:poisson}
-\hat\cL_0 \psi =u.
\end{equation}
The null space of $\hat\cL_0$, as well as its $ \mathcal{H}$-adjoint $\cL_0$ is one-dimensional and consists of constants. Consequently, the solvability condition for equations of the form \eqref{e:poisson} is that
\begin{equation}\label{e:solvability}
\langle {\bf 1}, u \rangle_\beta = 0.
\end{equation}
Provided that the solvability condition~\eqref{e:solvability} is satisfied, the Poisson equation~\eqref{e:poisson} has a unique mean zero solution, $\langle\psi , {\bf 1}\rangle_\beta = 0$. We correspondingly define the operator $ (-\hat{\mathcal{L}}_0)^{-1} $ on the subspace of functions
satisfying~\eqref{e:solvability} to be this unique mean zero solution.

From the first equation in~\eqref{e:f_j_eqns} and the normalization condition we deduce that
\begin{equation}
f_0 = 1. \label{e:f0}
\end{equation}
The properties of the operators $a^{\pm}$ immediately yield that the solvability condition is satisfied for all equations for $f_j, \, j=1, \dots$:
$$
\langle a^+ f_{j-1}, {\bf 1} \rangle_\beta = \langle f_{j-1}, a^- {\bf 1} \rangle_\beta = 0.
$$
The solution of Equations~\eqref{e:f_j} can be written as
$$
f_j = (-\hat\cL_0)^{-1} a^+ f_{j-1}= \hat\cK f_{j-1},
$$
where $\hat{\mathcal{K}}$ is the $\mathcal{H}$-adjoint of $\mathcal{K}:=a^{-}\left( -\mathcal{L}_0^{-1}\right)$. Consequently,
\begin{equation}\label{e:psi_soln}
f_j = \hat\cK^j {\bf 1}, \quad j=0,1, \dots
\end{equation}

Thus, we have obtained a power series expansion for the invariant distribution in powers of $F$:
\begin{equation}\label{e:rho_expansion}
\rho(p,q;F) = \bar{\rho}(p,q)  \left( 1 + \sum_{j \geq 1} F^j \hat\cK^j {\bf 1} \right)
\end{equation}
from which we immediately deduce the expansion for the effective drift:
\begin{eqnarray}
V &=&  \sum_{j \geq 1} F^j  \langle p, f_j \rangle_{\beta}  \nonumber \\ & = & 
\sum_{j \geq 1} F^j  \langle p, \hat \cK_j {\bf 1} \rangle_{\beta}  
\nonumber \\ &=&  
\beta^{-1}\sum_{j \geq 1} F^j  \langle a^+ {\bf 1}, \hat \cK^j {\bf 1} \rangle_{\beta} 
\nonumber \\ &=& 
\beta^{-1}\sum_{j \geq 1} F^j   \langle{\bf 1} , a^- \hat \cK^j {\bf 1} \rangle_{\beta}. \label{e:V_expansion}
\end{eqnarray}
In particular:
\begin{equation}\label{e:vj}
V_{j} = \beta^{-1} \langle{\bf 1} , a^- \hat \cK^j {\bf 1} \rangle_{\beta}.
\end{equation}

Now we proceed with the analysis of the Poisson equation~\eqref{e:cell} which, in view of~\eqref{e:V_expansion},~\eqref{e:f},~\eqref{e:fexp}, and~\eqref{e:f0} , can be written as 
\begin{equation}\label{e:cell_2}
- \cL \phi = p - \sum_{j \geq 1} F^j V_j, \quad \left\langle \phi, \sum_{j \geq 0} F^j f_j \right\rangle_{\beta}  = 0,
\end{equation}
with $V_j$ given by~\eqref{e:vj}. The generator of the perturbed dynamics is
$$
\cL = \cL_0 + F a^-,
$$
where $\cL_0$ is given by~\eqref{e:L0}. We look for a solution of~\eqref{e:cell_2} in the form of a power series expansion in $F$:
$$
\phi(p,q) = \sum_{j \geq 0} F^j \phi_j(p,q).
$$
We substitute this expansion into Equation~\eqref{e:cell_2} to obtain the sequence of equations (recalling from Eq.~\eqref{e:f0} that $ f_0 = 1 $):

\begin{subequations}\label{e:phi_j_eqns}
\begin{eqnarray}\label{e:phi_0}
-\cL_0 \phi_0 & = & p, \quad \langle  \phi_{0} , {\bf 1} \rangle_{\beta} = 0 \\
\label{e:phi_j}
-\cL_0 \phi_j & = & a^- \phi_{j-1} - V_j, \; \langle  \phi_{j} , {\bf 1} \rangle_{\beta} = 
- \sum_{r=1}^j \langle f_r, \phi_{j-r} \rangle_{\beta} \; j=1,2, \dots
\end{eqnarray}
\end{subequations}
Equation~\eqref{e:phi_0} is precisely the Poisson equation for the unperturbed dynamics $F = 0$. Now we show that the solvability condition~\eqref{e:solvability} is satisfied for equations~\eqref{e:phi_j}. We need to show that
\begin{equation}\label{e:solvab}
V_{j} = \langle a^-\phi_{j-1}, {\bf 1} \rangle_\beta.
\end{equation}
\begin{lemma}\label{lemma:vel}
The solvability condition~\eqref{e:solvab} is satisfied for all $ j \geq 1 $, and moreover the relation
\begin{equation} \label{eq:shiftphif}
\langle a^{\minus} \phi_{0}, f_{k} \rangle_{\beta} = 
\langle a^{\minus} \phi_{m}, f_{k-m} \rangle_{\beta} \text{ for }  k=m,m+1,\ldots
\end{equation}
holds for all $ m \geq 0 $.
%Furthermore, $f_j$ and $\phi_0$ are orthogonal:
%\begin{equation}\label{e:fjphi0}
%\langle \phi_0 , f_j \rangle_\beta = 0, \quad j=0,1, \dots
%\end{equation}
\end{lemma}

\begin{proof}
Our strategy pivots on the observation that if we can establish (\ref{eq:shiftphif})
for  $ m=0,1,2,\ldots,M $,
then the solvability condition~\eqref{e:solvab} follows for $ j =1,2,\ldots,M+1 $:
\begin{eqnarray*}
V_j &=& \beta^{-1} \langle \mathbf{1},a^{\minus} \hat{\mathcal{K}}^j \mathbf{1} \rangle_{\beta}
= \beta^{-1} \langle \mathcal{K} a^{\plus} \mathbf{1}, \hat{\mathcal{K}}^{j-1} \mathbf{1} \rangle_{\beta} \\
&=& \beta^{-1} \langle a^{\minus}  (- \mathcal{L}_0)^{-1} a^{\plus} \mathbf{1}, \hat{\mathcal{K}}^{j-1} \mathbf{1} \rangle_{\beta} = \langle a^{\minus} \phi_0, f_{j-1} \rangle_{\beta} 
=   \langle a^{\minus} \phi_{j-1},f_0 \rangle_{\beta} \\
&=& 
\langle a^{\minus} \phi_{j-1}, \mathbf{1} \rangle_{\beta}, 
\end{eqnarray*}
using Eq.~(\ref{eq:shiftphif}) with $ m= j-1 $ in the penultimate equality.  

We now proceed to establish (\ref{eq:shiftphif}) inductively for $ m=0,1,2,\ldots $.  
 The case $ m= 0 $ is trivial.  Suppose now Eq.~(\ref{eq:shiftphif}) has been shown for all  $ m=0,1,2,\ldots,M $; we will show that (\ref{eq:shiftphif}) also follows for $ m=M+1 $.  To this end, it is useful to introduce the operator $ \bar{\mathbb{P}} $ which projects orthogonally onto the hyperplane in $ \mathcal{H} $ orthogonal to constants:
\begin{equation*}
\bar{\mathbb{P}} g := g - \langle \boldsymbol{1}, g \rangle_{\beta}.
\end{equation*}
Since, by the above argument and the induction hypothesis, the solvability condition for \eqref{e:phi_j} is satisfied for $ j=1,2,\ldots,M+1$, we can write 
\begin{equation*}
  \phi_{M+1} = ( - \mathcal{L}_0)^{-1} \mathbb{P} a^{\minus} \phi_{M} -  \sum_{r=1}^{M+1} \langle f_r, \phi_{M+1-r} \rangle_{\beta},
  \end{equation*}
where the second sum of constants is included to meet the side condition in Eq.~\eqref{e:cell}, as we have defined $ (- \mathcal{L}_0)^{-1} $ to yield a mean zero solution.  But the operator $ a^{\minus} $ will kill these constants,   
  and therefore, for $ k=M+1,M+2,\ldots$, we can write:
\begin{eqnarray*}
\langle a^{\minus} \phi_{M+1}, f_{k-M-1} \rangle_{\beta}
&=& \langle a^{\minus} ( - \mathcal{L}_0)^{-1} \mathbb{P} a^{\minus} \phi_{M}, f_{k-M-1} \rangle_{\beta} \\
&=& \langle \mathcal{K} \mathbb{P} a^{\minus} \phi_{M},f_{k-M-1} \rangle_{\beta} 
= \langle a^{\minus} \phi_M, \mathbb{P} \hat{\mathcal{K}} f_{k-M-1} \rangle_{\beta} \\
&=& \langle a^{\minus} \phi_M, \hat{\mathcal{K}} f_{k-M-1} \rangle_{\beta}
= \langle a^{\minus} \phi_M,  f_{k-M} \rangle_{\beta} \\
& =& \langle a^{\minus} \phi_0,  f_{k} \rangle_{\beta}. 
\end{eqnarray*}
In the penultimate equality, we used $ \mathbb{P} \hat{\mathcal{K}} = \hat{\mathcal{K}} $ from the fact that by  definition of $(- \hat{\mathcal{L}}_0)^{-1}$, $ \mathbb{P} (- \hat{\mathcal{L}}_0)^{-1} = (- \hat{\mathcal{L}}_0)^{-1} $; the final equality follows from the induction hypothesis.
\end{proof}

Now we derive~\eqref{e:deff_coeffs1}. Using the centering condition in~\eqref{e:cell} we have that the diffusion coefficient is given by %(\textbf{PRK:  revised derivation. JCL: Upper index in last equation changed to r-1.  PRK:  Don't think so.  I restored my result and rewrote derivation to make it clearer.})
\begin{eqnarray*}
 D &=&  \int_{\bX} \! p \, \phi \, \rho_{\beta} \, dp dq
   \\ & = & 
    \sum_{\ell \geq 0} \sum_{n \geq 0} F^{n + \ell} \int_{\bX} \! p \, \phi_{\ell} \, f_{n} \bar{\rho} \, dp dq \\
    & = & 
    \sum_{r=0}^{\infty} \sum_{n=0}^{r} F^{r} \int_{\bX} \! p \, \phi_{r-n} \, f_{n} \bar{\rho} \, dp dq \\
& = & \sum_{r=0}^{\infty}F^r \left[\int_{\bX} \!  p \, \phi_{r} \, f_{0} \bar{\rho} \, dp dq
+ \sum_{n=1}^{r} \int_{\bX} \! p \, \phi_{r-n} \, f_{n} \bar{\rho} \, dp dq\right] \\
& = &  \sum_{r=0}^{\infty}  F^r \left[\beta^{-1} \langle \phi_r,a^{+} \boldsymbol{1}\rangle_{\beta}
+ \sum_{n=1}^{r}  \int_{\bX} \! p \, \phi_{r-n} \, f_{n} \bar{\rho} \, dp dq\right] \\
& = &  \sum_{r=0}^{\infty}  F^r \left[\beta^{-1} V_{r+1} + \sum_{n=1}^{r} \Sigma_{nr}\right],
\end{eqnarray*}
yielding (\ref{e:deff_coeffs1}). 
%    \\ & = &
%     \sum_{n \geq 0} F^{n} \int_{\bX} \! p \, \phi_{n} \, f_{0} \bar{\rho} \, dp dq + \sum_{n \geq 0} \sum_{\ell \geq 1} F^{n + \ell} \int_{\bX} \! p \, \phi_{n} \, f_{\ell} \bar{\rho} \, dp dq
%    \\
%     & =& 
%    \sum_{n \geq 0} F^{n} \beta^{-1} \langle a^{+} {\bf 1} , \phi_{n} \rangle_{\beta} + \sum_{r}^{n} \sum_{\ell \geq r} F^{\ell}  \int_{\bX} \! p \, \phi_{n} \, f_{\ell - n} \bar{\rho} dp dq
%    \\ & = & 
%    \beta^{-1} \sum_{\ell \geq 0} F^{\ell} V_{\ell+1} + \sum_{r =1 }^{n} \sum_{\ell \geq r} F^{\ell} \Sigma_{n \ell}. 

We can also alternatively restructure this expansion as follows, using the relations \eqref{e:solvab} and \eqref{eq:shiftphif}: 
\begin{eqnarray*}
D & = &  
    \sum_{r=0}^{\infty} \sum_{n=0}^{r} F^{r} \int_{\bX} \! p \, \phi_{r-n} \, f_{n} \bar{\rho} \, dp dq \\
  & = &  \sum_{r=0}^{\infty} \sum_{n=0}^{r} F^{r} \beta^{-1} \langle \phi_{r-n},(a^{+}+ a^{-}) f_{n} \rangle_{\beta} \\
    &= & \beta^{-1} \sum_{r=0}^{\infty} F^{r} \sum_{n=0}^{r} \left[\langle a^{-}\phi_{r-n}, f_{n} \rangle_{\beta} + 
    \langle \phi_{r-n},a^{-} f_{n} \rangle_{\beta}\right] \\
    &=& \beta^{-1} \sum_{r=0}^{\infty} F^r \sum_{n=0}^{r} \left[\langle a^{-}\phi_r, f_{0} \rangle_{\beta} + 
    \langle \phi_{r-n},a^{-} f_{n} \rangle_{\beta}\right]
     \\
    &= & \beta^{-1} \sum_{r=0}^{\infty} F^{r} 
    \sum_{n=0}^{r} \left[  V_{r+1} +   \langle \phi_{r-n},a^{-} f_{n} \rangle_{\beta}\right]
\\
 &= & \beta^{-1} \sum_{r=0}^{\infty} F^{r} 
   \left[(r+1) V_{r+1} +  \sum_{n=0}^{r}  \langle \phi_{r-n},a^{-} f_{n} \rangle_{\beta}\right] \\
    &= & \beta^{-1} \frac{d U}{d F} + 
     \sum_{r=1}^{\infty} F^r \sum_{n=1}^{r} \Xi_{nr},
         \end{eqnarray*}
establishing the statement~\eqref{e:deff_coeffs2} in the proposition.  In the last equality, we used that $ a^{-} f_0 = 0 $.
\qed

\begin{comment}
From the power series expansion in F of $\phi$ and $\rho$ we have,
\begin{eqnarray*}
 \int \! p \, \phi(q,p) \, \rho(q,p) \, \de p \de q &=& \int \! p \, \sum_{n=0}^\infty F^n \phi_n(q,p) \, \left( 1+ \sum_{n=1}^\infty F^n f_n(q,p) \right)\rho_\beta \, \de p \de q \\
&=& \sum_{n=0}^\infty F^n \beta^{-1} \int \! \beta p  \phi_n(q,p) \, \rho_\beta \, \de p \de q \\
&& + \int \! p \, \sum_{n=0}^\infty F^n \phi_n(q,p) \, \left( \sum_{n=1}^\infty F^n f_n(q,p) \right)\rho_\beta \, \de p \de q \\
&=& \sum_{n=0}^\infty F^n \beta^{-1}V_{n+1} + \sum_{n=0,m=1}^\infty \, F^{n+m} \int \! p \, \phi_n(q,p) \, f_m(q,p) \rho_\beta \de p \de q,
\end{eqnarray*}
where we have use the result from Proposition \ref{prop:vel} for the first term in the above expression. For the second term,, we re-organize the sum in terms of powers of $F$, which yields 
\begin{eqnarray*}
 D&=& \sum_{n=0}^\infty D_n F^n \\
&=& \sum_{n=0}^\infty F^n \beta^{-1}V_{n+1} + \sum_{n=1}^\infty \, F^{n} \sum_{j=1}^n \int \! p \, \phi_j(q,p) \, f_{n-j}(q,p) \rho_\beta \de p \de q,
\end{eqnarray*}
which proves equation (\ref{e:deff_n}).

\end{comment}

\begin{figure}[!ht]
\begin{center}
\includegraphics[width=0.95\textwidth]{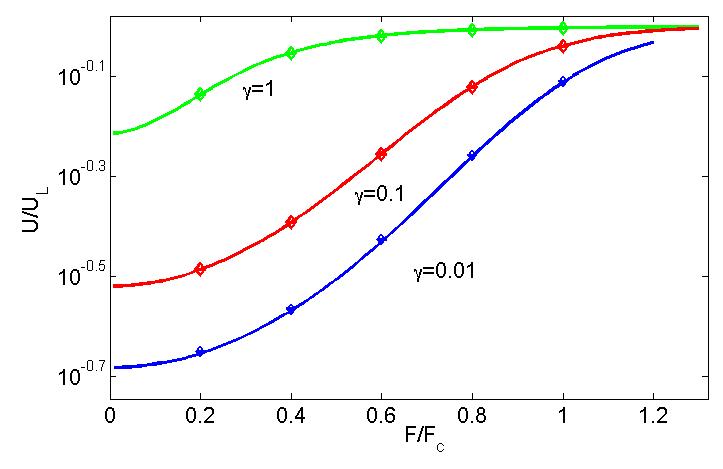}
\end{center}
\caption{\textbf{Solid Lines}: $U$ as a function of $F$ for $\gamma=0.01$, $0.1$ and $1$. \textbf{Markers}: Monte Carlo estimates. Other parameters of the simulations are $ V_0= \pi^2/16 $, $ \beta = 1.2 V_{0}^{-1} $, and $ L= 2 \pi $.  Following convention, the force variable is scaled by the critical force $F_c \approx 3.36 \gamma \sqrt{V_0} $ at which the effective potential~\eqref{e:effpot} becomes monotonic, and the drift is scaled by the value $ U_L = F/\gamma $ it would have in absence of the periodic potential $ \phi (x)$. For the Monte Carlo simulations we use an Euler-Maruyama scheme with a time step $\Delta t=0.1$ integrating over $N=5000000$ time steps (after $100000$ time steps of a transient integration interval) and averaging over $5000$ trajectories.} \label{f:u_small_gamma}
\end{figure}

\begin{figure}[!ht]
\begin{center}
\includegraphics[width=0.97\textwidth]{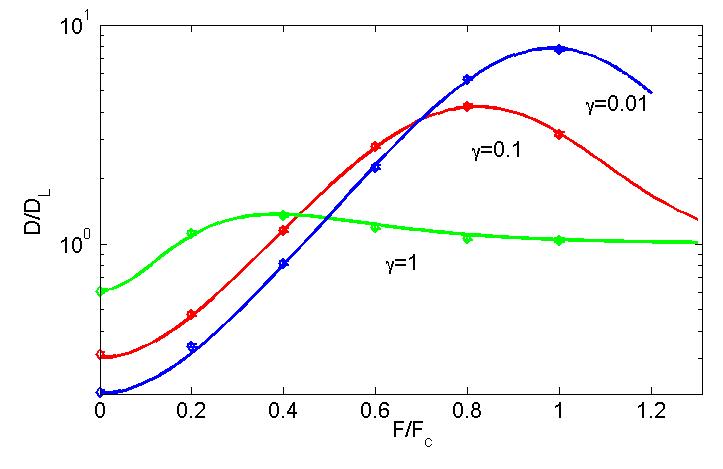}
\end{center}
\caption{\textbf{Solid Lines}: $D$ as a function of $F$ for  $\gamma=0.01$, $0.1$ and $1$. \textbf{Markers}: Monte Carlo estimates. Parameters of the simulations are as in Figure 1, with diffusivity now scaled by the value $ D_L = (\beta \gamma)^{-1} $  it would have in absence of the periodic potential $ \phi (x) $.   } \label{f:d_small_gamma}
\end{figure}

%%%%%%%%%%%%%%%%%%%%%%%%%%%%%%%%%%%%%%%%%%%%%%%%%%%%%%%%%%%
%           NUMERICAL SIMULATIONS
%                      
%
%%%%%%%%%%%%%%%%%%%%%%%%%%%%%%%%%%%%%%%%%%%%%%%%%%%%%%%%%%%
%
\section{Numerical Simulations}
\label{sec:numerics}

In this section we present results of numerical simulations that corroborate the theoretical results presented in the previous section. The calculation of the drift and diffusion coefficients is based on the numerical solution of the hypoelliptic boundary value problems~\eqref{e:fp_station} and~\eqref{e:cell} as well as the calculation of the integrals~\eqref{e:drift_defn} and~\eqref{e:deff}. Both PDEs  are solved using a spectral method that relies on the expansion of the solution of the stationary Fokker-Planck and the Poisson equations in a Fourier-Hermite expansion. This method is adapted from Risken's continued fraction expansion method~\cite{Ris84}; see also~\cite{FokGuoTang2002}. This method was used previously in the study of the diffusion coefficient for a Brownian particle in a periodic potential in~\cite{PavlVog08}. Details about the numerical method can be found in Appendix~\ref{sec:num_alg}.

\begin{figure}[!htp]
\centering
\begin{tabular}{c}
\includegraphics[width=0.85\textwidth]{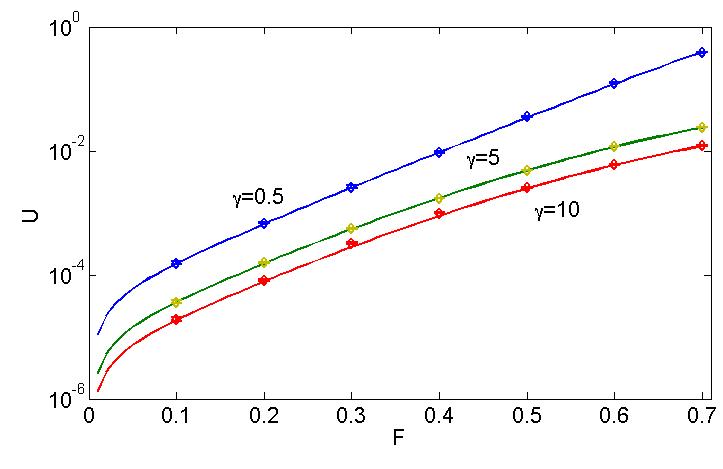}\\
a.~~ $U$ as a function of $F$ for  $\gamma=0.5, 5$ and $10$.\\
\includegraphics[width=0.85\textwidth]{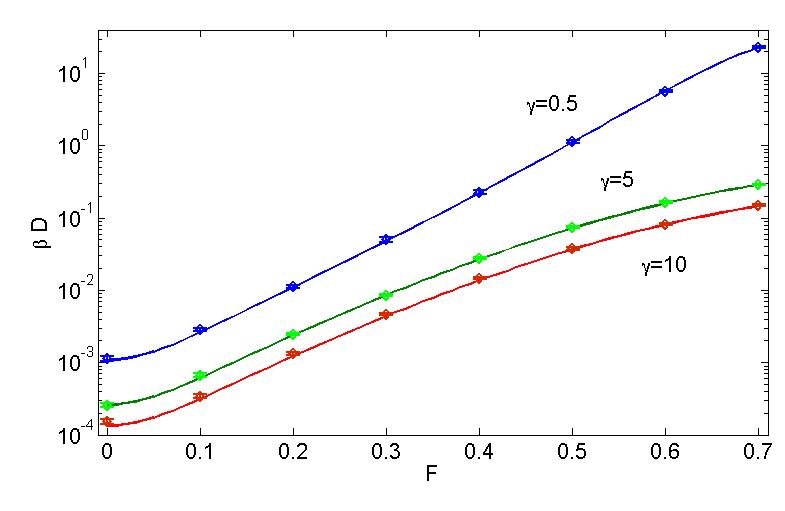} \\
  b.~~ $D$ as a function of $F$ for  $\gamma=0.5, 5$ and $10$.
\end{tabular}

  \caption{(a) \textbf{Solid lines}: $U$ computed from \eqref{e:drift_defn}. (b) \textbf{Solid lines}: $D$ computed from \eqref{e:deff}. \textbf{Markers}: Monte Carlo estimates. Parameters of the simulations are $ V_0= 1 $, $ \beta = 5$, and  $ L= 1$.  For the Monte Carlo simulations we have used a time step $\Delta t = 0.01$ over $6000000$ time steps and averaging over $6000$ trajectories.}\label{fig:linear_response}
  \end{figure}

 In all the numerical experiments we use a cosine potential, $V(q) = V_0\cos(\omega_1 q)$, with $\omega_1=2\pi/L$. %{\bf gp: DO WE ALSO CONSIDER COS(2 PI)?}. 
As a first test for the validity of our numerical method, in Figures \ref{f:u_small_gamma}, \ref{f:d_small_gamma} and \ref{fig:linear_response} we compare the results obtained from the solution of the two PDEs with results obtained using Monte Carlo simulations. In particular, in Figures \ref{f:u_small_gamma} and \ref{f:d_small_gamma} we reproduce the results reported in~\cite{CostantiniMarchesoni1999} for $\gamma=0.01$ and go beyond this for larger values of $\gamma$. In all the Monte Carlo simulations reported in this paper we take a sufficiently large number of realizations, a sufficiently small time step and sufficiently long paths so that the results of the simulations are very accurate.\footnote{In fact, in all the figures where the results of Monte Carlo simulations are presented, we also include the error bars. However, they are so small that they are barely visible.} Details on the values of the parameters used in the simulations can be found in the caption figures. In Figure \ref{fig:linear_response} we present results for $U$ and $D$ for larger values of $\gamma$.

We emphasize the fact that the spectral method enables us to calculate the drift and diffusion coefficients very accurately for a very wide range of values of the friction coefficient $\gamma$ as well as the forcing $F$. As expected, the numerical method becomes computationally more expensive as $\gamma$ decreases, since more Hermite and Fourier modes are needed for the accurate calculation of the diffusion coefficient. We note also that, in two and higher dimensions, 
the underdamped regime requires appropriate preconditioning for the efficient solution of the resulting linear algebraic problem.  
%(\textbf{PRK:  We need to say something about error bars of Monte Carlo simulations; are they so small they are on order of symbol size?})

\begin{figure}[!htp]
 \begin{center}
 \includegraphics[width=0.95\textwidth]{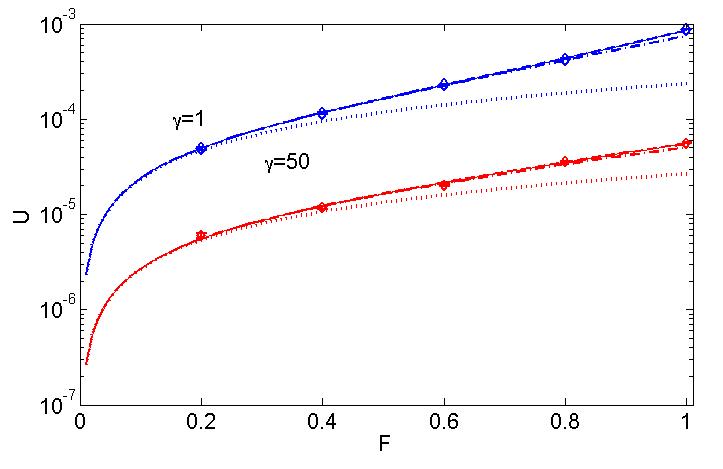}
\caption{$U$ for $\gamma=1$ (blue lines) and $\gamma=50$ (red lines). \textbf{Solid lines}: $U$ computed from \eqref{e:drift_defn}. \textbf{Markers}: Monte Carlo estimates. Broken lines: Expansion for $U$  in terms of $F$ by solving numerically for the coefficients $V_j$ given by \eqref{e:vel_j}. $\gamma=1$: \textbf{Dots}: linear approximation, \textbf{Dash}: 5th-order expansion, \textbf{Dash-Dash}: 9th-order expansion (overlapped with solid line). $\gamma=50$: \textbf{Dots}: linear approximation, \textbf{Dash-dot}: 3rd-order expansion, \textbf{Dash-Dash}: 5th-order expansion (overlapped with solid line).   Parameters of the simulations are $ V_0= 1 $, $ \beta = 5$, and  $ L= 1$. For the Monte Carlo simulations we have integrated $5000$ trajectories using a time step $\Delta t= 0.01$ over $1000000$ time steps for $\gamma = 1$, while $\Delta t =0.005$ over $4000000$ time steps for $\gamma=50$.} \label{u_linear_response}
\end{center}
\end{figure}

Now we turn our attention to the numerical study of formulas~\eqref{e:expansion1} and~\eqref{e:exeinstein}. In Figure~\ref{u_linear_response} we have calculated numerically the effective drift $U$ using~\eqref{e:drift_defn}, and we have also calculated numerically the coefficients $V_n$ in~\eqref{e:expansion1}. For this we need to solve the Poisson equations~\eqref{e:l*}, where the generator of the unperturbed dynamics, i.e. with $F=0$, appears. We can see that as we increase the number of terms in the power series expansion, the series converges to the value of $u$ computed from solving the stationary Fokker-Planck equation~\eqref{e:fp_station} and computing the integral in~\eqref{e:drift_defn}. We stress that, using the expansion~\eqref{e:expansion1} we can calculate the nonequilibrium drift for arbitrary values of the external forcing $F$ using only information from the equilibrium dynamics.

In Figure~\ref{d_linear_response} we plot the diffusion coefficient $D$, as a function of the forcing $F$, using~\eqref{e:deff}. In addition, we plot the power series expansions of different orders according to assumption~\eqref{e:exeinstein} that linear response relationships between drift and diffusion extend to the higher order coefficients. The drift $U$ is computed as in Figure~\ref{u_linear_response}, using the expansion~\eqref{e:expansion1}.  While the power series expansion does match the value of $ D $ for $ F=0 $, as it should according to linear response theory, the series does not converge to the values computed numerically for $ F\neq 0 $ using the spectral method described in Appendix~\ref{sec:num_alg}.
This shows in particular that the correction terms  $ \Xi_{k \ell} $ in Eq.~\eqref{e:deff_nj2} are nontrivial, as also evidenced by the results of Figure 1 in~\cite{reimann_al02}. To emphasize this point, we also compare in Figure~\ref{d_linear_responseNumericDerivative} the effective diffusion coefficient with formula~\eqref{e:exeinstein}, where $U$ is computed using~\eqref{e:drift_defn} and the derivative with respect to $F$ is approximated numerically using a centered finite difference scheme. See Figure~\ref{d_linear_responseNumericDerivative}. The linear response relations, however, do perform better for larger values of $ \gamma $.
\begin{figure}[!htp]
 \begin{center}
 \includegraphics[width=0.95\textwidth]{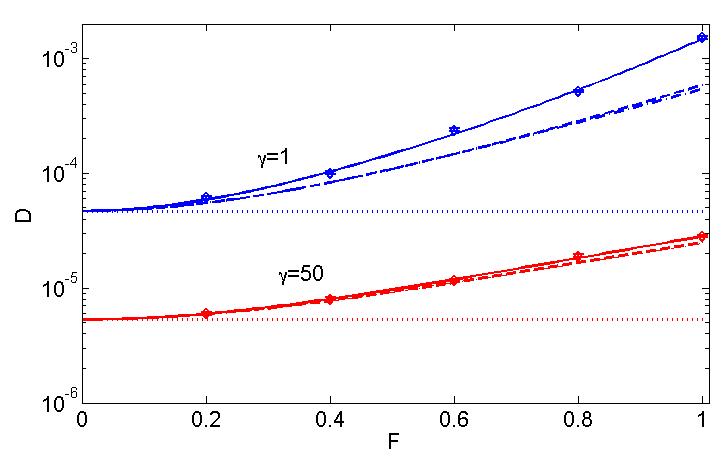}
\caption{$D$ for $\gamma=1$ (blue lines) and $\gamma=50$ (red lines). \textbf{Solid lines}:Homogenization formula~\eqref{e:deff}. \textbf{Markers}: Monte Carlo estimates. Broken lines: Expansion for $D$  in terms of $F$ assuming~\eqref{e:exeinstein}  and using the drift expansion~\eqref{e:expansion1}. $\gamma=1$:  \textbf{Dots}: constant approximation, \textbf{Dash}: 4th-order expansion, \textbf{Dash-Dot}: 8th-order expansion. $\gamma=50$: \textbf{Dots}: constant approximation, \textbf{Dash}: 2nd-order expansion, \textbf{Dash-Dot}: 6th-order expansion. Parameters of the simulation are as in Figure~\ref{u_linear_response}.} \label{d_linear_response}
\end{center}
\end{figure}
\begin{figure}[!htp]
 \begin{center}
 \includegraphics[width=0.95\textwidth]{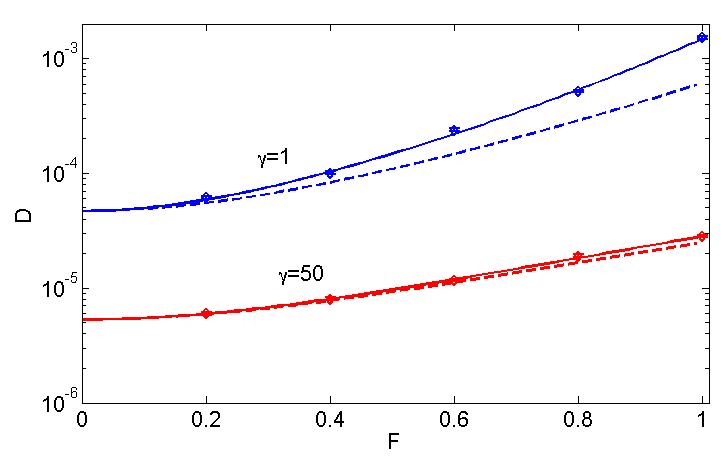}
\caption{$D$ for $\gamma=1$ (blue lines) and $\gamma=50$ (red lines).  \textbf{Solid lines}: Homogenization formula~\eqref{e:deff}. \textbf{Broken lines}: Computation of $D$ assuming the Einstein relation~\eqref{e:exeinstein}, with derivative evaluated through centered finite differences of the drift computed from the homogenization formula~\eqref{e:drift_defn}. Parameters of the simulation as in Figure~\ref{u_linear_response}.} \label{d_linear_responseNumericDerivative}
\end{center}
\end{figure}

Finally, we investigate the overdamped limit. The drift and diffusion coefficients of an overdamped particle moving in a one dimensional periodic potential under constant external force  can be computed analytically in terms of quadratures. The exact formula for the effective drift is computed in~(\cite{Straton58},\cite[Ch. 9]{Straton67}), whereas the exact formula for the diffusion coefficient can be found in~\cite{reimann_al02}  and~\cite{pavl05}.

Expanding~\eqref{e:fp_station} and~\eqref{e:cell} in inverse powers of $\gamma$ we obtain
\begin{equation} \label{e:v_overdamped_approx}
 U=\frac{U_{\mathrm{O}}}{\gamma}+O\left(\frac{1}{\gamma^3}\right)
\end{equation}
and
\begin{equation} \label{e:d_overdamped_approx}
 D=\frac{D_{\mathrm{O}}}{\gamma}+O\left(\frac{1}{\gamma^3}\right),
\end{equation}
where $U_{\mathrm{O}}$ and $D_{\mathrm{O}}$ denote the drift and diffusion coefficients for the overdamped problem (with $ \gamma $ scaled out, as described by the generator~\eqref{eq:overfpe} below). Rather than computing the integrals in the formulas for $U_{\mathrm{O}}$ and $D_{\mathrm{O}}$ (see, e.g.~\cite[Eqn. B.6 and Eqn. B.9]{pavl05}) we solve the stationary Fokker-Planck and the Poisson equations~\cite{PavlVog08} :
\begin{equation}
 U_{\mathrm{O}}=\int_0^L \!\! \left(-V'(q)+F\right) \rho_{\mathrm{O}}(q)\de q, \quad \cL_{\mathrm{O}}^*\rho_{\mathrm{O}}=0,
\end{equation}
with $\cL_{\mathrm{O}}^*$ the adjoint (Fokker-Planck operator) of the generator of the overdamped dynamics
\begin{equation}
 \cL_{\mathrm{O}}=\left(-V'(q)+F\right)\frac{\partial}{\partial q} + \beta^{-1}\frac{\partial^{2}}{\partial q^{2}}.
 \label{eq:overfpe}
\end{equation}
The generator is posed on $[0,L]$ equipped with periodic boundary conditions. Similarly, the diffusion coefficient is given by 
\begin{displaymath}
 D_{\mathrm{O}}=\beta^{-1}\int\!\! \left(1+\partial_q \phi_{\mathrm{O}}(q)\right)\rho_{\mathrm{O}}(q)\de q,
\end{displaymath}
with
\begin{displaymath}
- \cL_{\mathrm{O}}\phi_{\mathrm{O}}(q)=-V'(q)+F,
\end{displaymath}
on $[0,L]$ with periodic boundary conditions. Higher order corrections in~\eqref{e:v_overdamped_approx} and~\eqref{e:d_overdamped_approx} can be obtained through the solution of further auxiliary Poisson equations.   As shown in Figure~\ref{fig:overdamped_approx} , the overdamped formulas for the drift and diffusion coefficients offer a very accurate approximation even for moderately high values of the friction coefficient, uniformly in $F$.  

\begin{figure}
\centering
\begin{tabular}{c}
\includegraphics[width=0.85\textwidth]{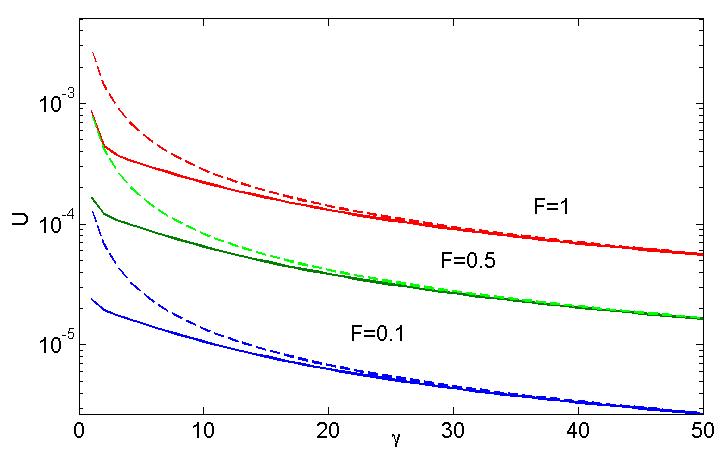}\\
a.~~ $U$ as a function of $\gamma$.\\
\includegraphics[width=0.85\textwidth]{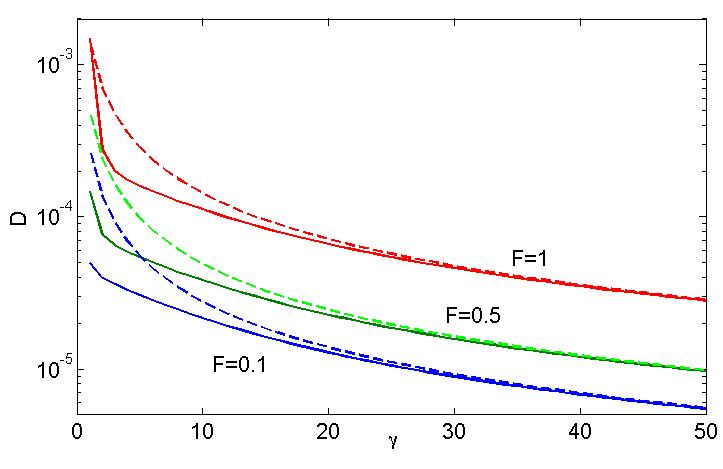} \\
  b.~~ $D$ as a function of $\gamma$.
\end{tabular}

\label{fig:overdamped_approx_d}   \caption{(a) \textbf{Solid lines}: Effective drift $U$ computed from~\eqref{e:drift_defn}. \textbf{Broken lines}: Asymptotic approximation (\ref{e:v_overdamped_approx}). (b) \textbf{Solid lines}: Effective diffusivity $D$  computed from~\eqref{e:deff}. \textbf{Broken lines}: Asymptotic approximation (\ref{e:d_overdamped_approx}).  Parameters of the simulations are $V_0=1$,  $\beta=5$, and $L=1$. }\label{fig:overdamped_approx}
  \end{figure}

%%%%%%%%%%%%%%%%%%%%%%%%%%%%%%%%%%%%%%%%%%%%%%%%%%%%%%%%%%%%%%%%%%%%%%%%%%
%
%
%
%%%%%%%%%%%%%%%%%%%%%%%%%%%%%%%%%%%%%%%%%%%%%%%%%%%%%%%%%%%%%%%%%%%%%%%%%%%
%
\section{Conclusions}
\label{sec:conclusions}

Using the framework of homogenization theory and multiscale analysis, we have developed a systematic expansion of the effective drift and effective diffusivity for the nonequilibrium dynamics of a particle in a tilted periodic potential.  The coefficients in this expansion relate the nonequilibrium transport coefficients to statistical averages involving the equilibrium dynamics (with no imposed tilt), computed through the solutions of boundary value problems for deterministic partial differential equations of hypoelliptic type.   The expansions give a detailed description of how Einstein's relation between the diffusivity and mobility of a particle is violated in higher orders with respect to the perturbation from equilibrium.
Our theoretical results were confirmed by numerical simulations based on a new  efficient spectral method for the solution of Poisson equations for the generator of the Langevin dynamics.

Our method of analysis can be readily extended, with suitable elaboration of notation, to multiple dimensions. Other substantial directions for future research include the application of the homogenization procedure to 
multiscale and locally periodic potentials, as well as to time-dependent external forcing.  This last setting could have particular relevance to the study of stochastic resonance phenomena.

\bigskip

\noindent{\bf Acknowledgments} 
This work is supported by the DFG Research Center {\sc Matheon} ``Mathematics for Key Technologies'' (FZT86) in Berlin. The research of G.A.P. is partially supported by the EPSRC, Grant No. EP/H034587.   PRK wishes to thank the Zentrum f\"{u}r Interdisziplin\"{a}re Forschung (ZiF) for its hospitality and support during its ``Stochastic Dynamics:  Mathematical Theory and Applications" program, during which part of this research was completed.  The research of J.C.L. was partially supported by NSF DMS-0449717.
%
%
%%%%%%%%%%%%%%%%%%%%%%%%%%%%%%%%%%%%%%%%%%%%%%%%%
%
%
\appendix

\section{Numerical Algorithm}
\label{sec:num_alg}
In this appendix we present a numerical approach for solving the 1-dimensional stationary Fokker-Planck equation \eqref{e:fp_station} together with the cell problem \eqref{e:cell} for computing $V$ and $D$ via \eqref{e:drift_defn} and \eqref{e:deff}. This numerical method is based on the approach by \cite{Ris84} and consists in a spectral decomposition of the solution of (\ref{e:fp_station}) and \eqref{e:cell} in terms of Hermite polynomials and Fourier series, followed by a recursive method to solve the resulting system of algebraic equations. Since this approach is presented in \cite{Ris84} for finding numerically $\rho_\beta(q,p)$ and $U$, we will focus on the computation of $D$ via the solution for the auxiliary field $\phi(q,p)$ in equation (\eqref{e:cell}) and equation (\ref{e:deff}).
\subsection{Solution in terms of Hermite polynomials.}
The cell problem for the auxiliary field $\phi(q,p)$  can be written in terms of the infinitesimal generator of the Ornstein-Uhlenbeck (OU) process as introduced in Section \ref{sec:lin_resp},
\begin{displaymath}
\cS=-p\partial_p+\beta^{-1}\partial^{2}_{p},
\end{displaymath}
as
\begin{displaymath}
\cL\phi(q,p)=p\partial_q\phi+(-V'(q)+F)\partial_p\phi+\gamma \cS\phi=U-p,
\end{displaymath}
with $U$ the effective drift as given by \eqref{e:drift_defn}. Note that the invariant distribution $\hat{\rho}(p)$ of the OU process, $\cS^*\hat{\rho}(p)=0$, implies $\hat{\rho} \sim \EXP{-\beta p^2/2}$. In view of the structure of the previous equation, we use the following representation for its solution,
\begin{equation}\label{sol_cell}
\phi(q,p)=\sum_{n=0}^\infty \, \phi_n(q)H_n(p),
\end{equation}
where $\phi_n(q)$ is a series of functions to be determined. $H_n(p)$ are rescaled Hermite polynomials 
 \begin{eqnarray*}
H_n (p) &=& \frac{1}{\sqrt{n!}} \mathrm{He}_n (p \beta^{1/2}), \\ 
\mathrm{He}_n (x) &=& (-1)^n e^{x^2/2} \frac{d^n}{dx^n} e^{-x^2/2},
\end{eqnarray*} 
which are the eigenfunctions of the operator $\cS$,
\begin{displaymath}
\cS H_n(p)=-nH_n(p), \quad n=1,2,\ldots
\end{displaymath}
Also, these rescaled Hermite polynomials are orthonormal with respect to the unperturbed stationary distribution:
\begin{displaymath}
\langle H_n(p) H_m(p)\rangle_{\hat{\rho}}=\int_{-\infty}^\infty H_n(p) H_m(p) \bar{\rho}(p) \, \dy =\delta_{nm},
\end{displaymath}
and satisfy the  relations :
\begin{eqnarray*}
p H_n (p) &=& \beta^{-1/2}\left(\sqrt{n+1} H_{n+1} (p) + \sqrt{n} H_{n-1} (p)\right), \\
H_n^{\prime} (p) &=& (\beta n)^{1/2} H_{n-1} (p).
\end{eqnarray*}
Upon substituting (\ref{sol_cell}) into (\eqref{e:cell}), projecting against $ H_0 $, $H_1 $, and $H_n $ for $ n \geq 2 $ respectively , and using the orthonormality property of the Hermite polynomials, we obtain the following infinite system of ordinary differential equations for $\phi_n(q)$, 
%(\textbf{PRK:  Corrected some factors; please check. These factors emerge from the Hermite polynomial relations I have written above.})
\begin{displaymath}
\begin{array}{rl}
\phi_1'(q)+\beta\,(-V'(q)+F)\phi_1(q)&=\ds \sqrt{\beta}\,U, \vspace{7pt} \\
\ds \phi_0'(q)-\gamma\ds\sqrt{\beta}\,\phi_1(q)+\beta\left(-V'(q)+F\right) \sqrt{2} \,\phi_2(q) +\sqrt{2}\,\phi_2'(q)&=-1, \vspace{7pt} \\
\sqrt{n}\,\phi_{n-1}'(q)-\gamma\ds\sqrt{\beta}\,n \,\phi_n(q) +\sqrt{n+1}\,\phi_{n+1}'(q)& \vspace{7pt} \\
\qquad \qquad +\beta\left(-V'(q)+F\right) \sqrt{n+1}\,\phi_{n+1}(q) &=0,
\end{array}
\end{displaymath}
for $n=2,3, \ldots$
\subsection{Spectral decomposition.}
Since the solution to the cell problem must be periodic in $q$, the auxiliary functions $\phi_n(q)$ must also be periodic. It is natural then to express these functions in terms of their Fourier series,
\begin{displaymath}
\phi_n(q)=\sum_{j=-\infty}^\infty \, \Phi_n^j \, \EXP{i \omega_j q}, \quad \omega_j=\frac{2\pi j}{L}.
\end{displaymath}
For simplicity, we will focus now on the simplest periodic potential, namely, $$V(q)=V_0\rm{cos}(2\pi q /L),$$ although more complex potentials can be studied. In terms of this potential, the equations take the following form,
\begin{eqnarray}\label{spectral}
i\omega_j \Phi_1^j+\beta\,\left(\frac{V_0\omega_1}{2i}\left(\Phi^{j-1}_1-\Phi^{j+1}_1\right)+F\Phi^j_1\right)&=&\sqrt{\beta}\,U\delta_{j,0}, \notag \vspace{7pt}\\
i\,\omega_j \Phi^j_0-\gamma\sqrt{\beta}\,\Phi^j_1+ && \vspace{18pt} \notag \\
\sqrt{2}\left(i\omega_j\,\Phi^j_2+\beta\,\left(\frac{V_0\omega_1}{2i}\left(\Phi^{j-1}_2-\Phi^{j+1}_{2}\right)+F\Phi_{2}^j\right)\right)&=&-\delta_{j,0}, \notag \vspace{7pt} \\
i\sqrt{n}\,\omega_j \Phi_{n-1}^j-\gamma\sqrt{\beta}\,n\,\Phi_n^j+ && \vspace{18pt} \notag  \\
\sqrt{n+1}\left(i\omega_j\,\Phi_{n+1}^j+\beta\left(\frac{V_0\omega_1}{2i}\left(\Phi_{n+1}^{j-1}-\Phi_{n+1}^{j+1}\right)+F\Phi_{n+1}^j\right)\right)&=&0,\\
\textrm{for } n=2,3,\ldots, \quad j=0,\pm 1, \pm 2, \ldots && \notag
\end{eqnarray}
\subsection{Solution of $\phi_n$.}
We now proceed to describe the numerical algorithm for computing $D$. In order to solve \eqref{e:cell} in its spectral representation (\ref{spectral}), we approximate $\phi_n(q)$ by a Galerkin  truncation of the Fourier series after the $M$th  term,
\begin{displaymath}
\phi_n(q) \approx \sum_{j=-M}^M \, \Phi_n^j \, \EXP{i \omega_j q}.
\end{displaymath}
The infinite system of algebraic equations (\ref{spectral}) becomes then an infinite, tri-diagonal system of equations expressed as follows. 
By explicitly writing the real and imaginary parts of $\Phi_n^j=\xi_n^j+i \eta_n^j$ and using the fact that the solution must be real-valued (which implies that $\xi_{n}^{-j}=\xi_n^j$ and $\eta_{n}^{-j}=-\eta_n^j$ ) we form the vectors,
\begin{displaymath}
\boldsymbol{\Phi}_n=\left( \begin{array}{c}
\xi_n^0 \vspace{6pt}\\
\xi_n^1 \\
\vdots \\
\xi_n^M\vspace{6pt}\\
\eta_n^1 \\
\vdots \\
\eta_n^M
\end{array}\right), \quad n=0,1,2,\ldots 
\end{displaymath}
This representation leads to the following system of equations,
\begin{subequations}
\begin{eqnarray}
\boldsymbol{Q}_0^-\boldsymbol{\Phi}_{1}&=&\boldsymbol{B}, \vspace{7pt} \label{eq:qno} \\
\boldsymbol{Q}_1^+\boldsymbol{\Phi}_{0}+\boldsymbol{Q}_1\boldsymbol{\Phi}_{1}+\boldsymbol{Q}_1^-\boldsymbol{\Phi}_{2}&=&\boldsymbol{A},\vspace{7pt} \label{eq:qn1}\\
\ds \boldsymbol{Q}_n^+\boldsymbol{\Phi}_{n-1}+\boldsymbol{Q}_n\boldsymbol{\Phi}_{n}+\boldsymbol{Q}_n^-\boldsymbol{\Phi}_{n+1}&=&\boldsymbol{0}, \quad n=2,3,\ldots \label{eq:qn2}
\end{eqnarray}
\end{subequations}
These matrices are given, for $n=0,1,\ldots$, by,
\begin{displaymath}
\boldsymbol{Q}_n=-\gamma\, \sqrt{\beta}\, n \boldsymbol{I}_{2M+1},
\end{displaymath}
where $\boldsymbol{I}_{k}$ is the $k$ x $k$ identity matrix. For $n=1,2,\ldots$ we have,
\begin{displaymath}
\boldsymbol{Q}_n^+=\sqrt{n} \left( \begin{array}{cccccc|cccc}
 &           &      &  & & & 0&0     & \ldots   &    0       \vspace{7pt} \\              
 &          &   &        &      & &  -\omega_1 &0&\ldots       &  0      \vspace{7pt} \\
   &           & & \boldsymbol{0}&& &  \vdots         &\ddots & &\vdots       \vspace{7pt} \\
&       &    &       &     &  & \ldots  & 0&      & -\omega_M  \vspace{7pt} \\
\hline
0 & \omega_1 &   0   &0& \ldots& 0&&&                                  &    \vspace{7pt} \\       
0& 0  &\omega_2  &0    &\ldots & 0&&\boldsymbol{0}&                                  &          \vspace{7pt} \\
\vdots &&\ddots&& \vdots &&&&& \vspace{7pt} \\
 0&  0     &   & \ldots& 0 & \omega_M &&&                                  &
\end{array} \right), \vspace{10pt}
\end{displaymath}
\begin{displaymath}
\boldsymbol{Q}_n^-=\sqrt{n+1}\left(\begin{array}{c|c}
\boldsymbol{Q}_{aa} & \boldsymbol{Q}_{ab} \vspace{6pt}\\
\hline \vspace{-7pt}\\ 
\boldsymbol{Q}_{ba} & \boldsymbol{Q}_{bb}
\end{array}\right), \vspace{10pt}
\end{displaymath}
\begin{displaymath}
\boldsymbol{Q}_{aa}=F\, \beta\, \boldsymbol{I}_{M+1}, \quad \boldsymbol{Q}_{bb}=F\, \beta\, \boldsymbol{I}_{M}
\end{displaymath}
\begin{displaymath}
\boldsymbol{Q}_{ab}=\left(\begin{array}{cccccc}
-\beta\,V_o\,\omega_1 & 0 &0&0&\ldots &0 \vspace{6pt}\\
 -\omega_1 & -\beta\,V_o\,\omega_1/2 &0&0&\ldots &0 \vspace{6pt}\\
\beta\,V_o\,\omega_1/2 & -\omega_2 & -\beta\,V_o\,\omega_1/2 &0&\ldots&0 \\
\ddots&\ddots &\ddots& & &\vdots \vspace{20pt}\\
0 &\ldots &&0&\beta\,V_o\,\omega_1/2 & -\omega_M 
\end{array}\right), \vspace{10pt}
\end{displaymath}
\begin{displaymath}
\boldsymbol{Q}_{ba}=\left(\begin{array}{ccccccc}
-\beta\,V_o\,\omega_1/2 & \omega_1 & \beta\,V_o\,\omega_1/2 &0&0&\ldots &0 \vspace{6pt}\\
0 & -\beta\,V_o\,\omega_1/2 & \omega_2 & \beta\,V_o\,\omega_1/2 &0 &\ldots&0 \vspace{6pt}\\
\vdots &\ddots&\ddots &\ddots&\vdots&&\vdots \vspace{20pt}\\
0 &\ldots && &0&-\beta\,V_o\,\omega_1/2 & \omega_M 
\end{array}\right). \vspace{10pt}
\end{displaymath}
\begin{displaymath}
[\boldsymbol{B}]_k=\sqrt{\beta}\,U\delta_{k,1}, \quad [\boldsymbol{A}]_k=-\delta_{k,1},
\end{displaymath}
where $[\boldsymbol{B}]_k$ represents the $k$th  element of the vector $\boldsymbol{B}$ (respectively for $\boldsymbol{A}$.)  In order to solve the infinite system of algebraic equations, we impose some boundary condition of the form $\boldsymbol{\Phi}_{N+1}=\boldsymbol{S}_N\boldsymbol{\Phi}_N$, for large $N$. Tested boundary conditions include $\boldsymbol{S}_N=\boldsymbol{0}$ (Dirichlet boundary condition), which we employed in the simulations in Section~\ref{sec:numerics} , and $\boldsymbol{S}_N=\boldsymbol{I}_{2M+1}$ (Neumann boundary condition). Defining matrices $\{ \boldsymbol{S}_n\}_{n=0}^{N-1} $ recursively downwards from $ n = N-1 $ by
\begin{equation*}
\boldsymbol{S}_{n} = - (\boldsymbol{Q}_{n+1}+ \boldsymbol{Q}_{n+1}^{-} \boldsymbol{S}_{n+1})^{-1} \boldsymbol{Q}_{n+1}^{+}
\text{ for } n=0,\ldots,N-1,
\end{equation*}
we can check by induction (again downwards)  that for $ n=1,\ldots,N$,
\begin{equation}
\boldsymbol{\Phi}_{n+1} = \boldsymbol{S}_{n} \boldsymbol{\Phi}_{n} \label{e:phisinduct}
\end{equation}
Indeed, this relation is already in force for $ n=N $, and assuming it to be true for some $ n=m \geq 2$, from Eq.~\eqref{eq:qn2} we find: 
\begin{displaymath}
\boldsymbol{Q}_m^+\boldsymbol{\Phi}_{m-1}+\boldsymbol{Q}_m\boldsymbol{\Phi}_{m}+\boldsymbol{Q}_m^-\boldsymbol{S}_m\boldsymbol{\Phi}_{m}=\boldsymbol{Q}_m^+\boldsymbol{\Phi}_{m-1}+\left(\boldsymbol{Q}_m+\boldsymbol{Q}_m^-\boldsymbol{S}_m\right)\boldsymbol{\Phi}_{m}=0, 
\end{displaymath} 
so that Eq.~\eqref{e:phisinduct} holds for $ n=m-1 $ as well.  
%The system of equations is then solved recursively downwards from $n=N+1$ to $n=1$, yielding matrices satisfying the relations $ \boldsymbol{\Phi}_{n}
%= \boldsymbol{S}_{n-1} \boldsymbol{\Phi}_{n-1} $.  We begin by writing,
%\begin{displaymath}
%\boldsymbol{Q}_N^+\boldsymbol{\Phi}_{N-1}+\boldsymbol{Q}_N\boldsymbol{\Phi}_{N}+\boldsymbol{Q}_N^-\boldsymbol{S}_N\boldsymbol{\Phi}_{N}=\boldsymbol{Q}_N^+\boldsymbol{\Phi}_{N-1}+\left(\boldsymbol{Q}_N+\boldsymbol{Q}_N^-\boldsymbol{S}_N\right)\boldsymbol{\Phi}_{N}=0. 
%\end{displaymath} 
%Upon writing,
%\begin{displaymath}
%\boldsymbol{\Phi}_{N}=\left(\boldsymbol{Q}_N+\boldsymbol{Q}_N^-\boldsymbol{S}_N\right)^{-1}\boldsymbol{Q}_N^+\boldsymbol{\Phi}_{N-1}=\boldsymbol{S}_{N-1}\boldsymbol{\Phi}_{N-1},
%\end{displaymath}
%we can find recursively $\boldsymbol{S}_{n}$, $n=N-1,N-2,\ldots,2$. 
Turning now to Eqs.~\eqref{eq:qno} and~\eqref{eq:qn1}, we have
\begin{subequations}
\begin{eqnarray}
\boldsymbol{Q}_1^+\boldsymbol{\Phi}_0+\left(\boldsymbol{Q}_1+\boldsymbol{Q}_1^-\boldsymbol{S}_1\right)\boldsymbol{\Phi}_1&=&\boldsymbol{A} \vspace{7pt}  \\
\boldsymbol{Q}_0^-\boldsymbol{\Phi}_1&=&\boldsymbol{B}, \label{eq:phi1s}
\end{eqnarray}
\end{subequations}
from which we find by solving for $ \boldsymbol{\Phi}_1 $ in terms of $ \boldsymbol{\Phi}_0 $:
\begin{equation*}
\boldsymbol{\Phi}_1 = \boldsymbol{S}_0 \boldsymbol{\Phi}_0 +\left(\boldsymbol{Q}_1+\boldsymbol{Q}_1^-\boldsymbol{S}_1\right)^{-1} \boldsymbol{A}.
\end{equation*}
Substituting this expression into Eq.~\eqref{eq:phi1s}, we finally obtain a closed equation for $ \boldsymbol{\Phi}_0 $:
\begin{equation}
\boldsymbol{Q}_0^- \boldsymbol{S}_0\boldsymbol{\Phi}_0=\boldsymbol{B} - \boldsymbol{Q}_0^-\left(\boldsymbol{Q}_1+\boldsymbol{Q}_1^-\boldsymbol{S}_1\right)^{-1}\,\boldsymbol{A}. \label{e:phios}
\end{equation}
 The matrix $\boldsymbol{Q}_0^-\boldsymbol{S}_0$  will have one null eigenvalue (corresponding to the null space of $\cL$).   One can verify, by considering the analogous numerical solution scheme for $ \rho_{\beta} $ and $ U $, presented in~\cite{Ris84}, that the right hand side of Eq.~\eqref{e:phios} satisfies the solvability condition that it be orthogonal to the left eigenvector of $ \boldsymbol{Q}_0^- \boldsymbol{S}_0 $ with zero eigenvalue.  A unique solution for $ \boldsymbol{\Phi}_0 $ is then obtained by discretization of the auxiliary condition in Eq.~\eqref{e:cell}.  In particular, representing the solution to the stationary Fokker-Planck (\ref{e:fp_station}) by a Hermite polynomial expansion
\begin{displaymath}
\rho_{\beta}(q,p)=\hat{\rho}(p)\sum_{n=0}^\infty R_n(q)H_n(p),
\end{displaymath}
and approximating the functions $ R_n (q) $ by a finite Fourier series, with coefficients organized into vectors $ \boldsymbol{R}_n $ analogously to Eq.~\eqref{sol_cell}, this auxiliary condition reads:
\begin{equation*}
\sum_{n=0}^{N} \Big(2 \boldsymbol{R}_n^{T} \boldsymbol{\Phi}_n - \left[\boldsymbol{R}_n^{T} \boldsymbol{\Phi}_n\right]_1 \Big)= 0.
\end{equation*}
This then determines, with Eq.~\eqref{e:phios},  $\boldsymbol{\Phi}_0$  from which he remaining $\{\boldsymbol{\Phi}_n\}_{n=1}^N$ are found recursively using the matrices $\boldsymbol{S}_n$ and the relations~\eqref{e:phisinduct}. Once the vectors $\boldsymbol{\Phi}_{i}$ are found, $D$ is easily computed by replacing the proposed solution for $\rho_{\beta}(q,p)$ and $\phi(q,p)$ in (\ref{e:deff}) and using the Hermite polynomial properties to obtain:
\begin{displaymath}
\begin{array}{rl}
D=& \ds L\sum_{n=0}^N \, \sqrt{\frac{n+1}{\beta}}\Big(2\boldsymbol{R}_{n+1}^T\boldsymbol{\Phi}_{n}-[\boldsymbol{R}_{n+1}^T\boldsymbol{\Phi}_{n}]_1+2\boldsymbol{R}_{n}^T\boldsymbol{\Phi}_{n+1}  -[\boldsymbol{R}_{n}^T\boldsymbol{\Phi}_{n+1}]_1\Big)
% \vspace{12pt} \\
%& \ds -U \left(2\boldsymbol{R}_{n}^T\boldsymbol{\Phi}_{n}-[\boldsymbol{R}_{n}^T\boldsymbol{\Phi}_{n}]_1\right)\bigg)
\end{array}
\end{displaymath}

\section{Alternative Approach to Obtaining Corrections to Einstein's Formula}
\label{sec:appmaes}
The relation between the diffusivity and mobility is expressed in~\cite{BaiesiMaes11} as follows (in our notation):
\begin{equation}
D = \beta^{-1} \frac{dU}{dF} +\lim_{T \rightarrow  \infty} \frac{1}{2 \gamma} 
\int_0^T \left\langle \frac{(q (T) - q (0)}{T}; - V^{\prime} (q(t)) + F \right\rangle \, d t ,
\label{e:maes}
\end{equation}
where $ \langle g ; h  \rangle = \langle g h \rangle - \langle g \rangle \langle h \rangle $ and $ \langle \cdot \rangle $  denotes an average over the stochastic noise (and possibly random initial conditions).  The correction term was studied in~\cite{BaiesiMaes11} on the model system (\ref{e:langevin_intro}) as well as other non-equilibrium systems through direct numerical simulation of the governing dynamical equations and Monte Carlo estimation of the statistical average.  We can express Eq.~\eqref{e:maes} in terms of deterministic operators through the following formal manipulations.  First, we re-express
\begin{equation*}
q(T) - q(0) = \int_0^T p(\tp) \, d \tp,
\end{equation*}
which avoids the complication of working with the nonperiodic variable $ q(t) $.  We then have:
\begin{equation*}
D = \beta^{-1} \frac{dU}{dF} +\lim_{T \rightarrow  \infty} 
\frac{1}{2 \gamma T} 
\int_0^T \int_0^T \left\langle p(\tp); - V^{\prime} (q(t)) + F \right\rangle \, d \tp \, d t
\end{equation*}
Now, thanks to the large factor of $ T $ in the denominator, we may neglect initial transients and evaluate the statistical average in the nonequilibrium steady state, i.e., with single-time statistics governed by the invariant density $ \rho_{\beta}$, the solution of the stationary Fokker-Planck equation \eqref{e:fp_station}.   We then express the two-time correlation function formally using the evolution operator $ e^{\cL \Delta t} $, where $\cL$ denotes the generator of the Langevin dynamics, for the forward-in-time variable, and the projection operator 
\begin{equation*}
\mathbb{P} g =  g - \langle g \rangle_{\rho}
\end{equation*}
where
\begin{equation*}
\langle g \rangle_{\rho} \equiv \int_{\bX} g \rho_{\beta} \, dp \, dq
\end{equation*}
to obtain:
\begin{eqnarray*}
D &=& \beta^{-1} \frac{dU}{dF} - \lim_{T \rightarrow  \infty} 
\frac{1}{2 \gamma T} 
\int_0^T \left[\int_0^t  \left\langle \left(\mathbb{P} p\right) \left(e^{\cL (t-\tp)} \mathbb{P} 
 V^{\prime} (q) \right) \right\rangle_{\rho} \, d \tp \right. \\
& & \qquad  \qquad \left. + \int_t^T \left\langle \left(e^{\cL (\tp-t)} \mathbb{P} p\right) \left( 
 \mathbb{P} V^{\prime} (q) \right) \right\rangle_{\rho} \, d \tp\right] \, dt \\
& &= \beta^{-1} \frac{dU}{dF} - \lim_{T \rightarrow  \infty} 
\frac{1}{2 \gamma T} 
\left[ \left\langle \left(\cL^{-2} (e^{\cL T} - \Id - \cL T) \mathbb{P} p\right) \left(\mathbb{P} 
 V^{\prime} (q)\right) \right\rangle_{\rho} \right. \\
& & \qquad \qquad  \left. +  \left\langle \left( \mathbb{P} p\right) \left( \cL^{-2} (e^{\cL T} - \Id - \cL T)
 \mathbb{P} V^{\prime} (q)\right) \right\rangle_{\rho} \right],
  \end{eqnarray*}
  where $ \Id $ is the identity operator.
Using now the nonpositivity of the operator $ \cL $   and the fact that $ \Ran \mathbb{P} = \Ran \cL $ since $ \mathbb{P} $ projects onto the subspace orthogonal to  the kernel of $ \cL^{*} $, the $ L^2 $ adjoint of $ \cL $, we can evaluate the $ T \rightarrow \infty $ limit to obtain the following formal operator-theoretic equivalent to the formula Eq.~\eqref{e:maes} from~\cite{BaiesiMaes11}:
\begin{equation}
D =  \beta^{-1} \frac{dU}{dF} + \frac{1}{2 \gamma}
\left[ 
  \left\langle \left(\cL^{-1} \mathbb{P} p\right) 
 V^{\prime} (q) \right\rangle_{\rho} 
 +  \left\langle \left( \mathbb{P} p\right) \left( \cL^{-1}
 \mathbb{P} V^{\prime} (q)\right) \right\rangle_{\rho} \right].
 \label{e:maesdetlong}
\end{equation}
A somewhat more compact formula can be obtained by defining the $ L^2 (\rho) $ adjoint of $ \cL$, which can be computed as:
\begin{equation}
\hat{\cL} = \hat{\cL}_0 - F a^{\minus} + 2 \gamma (p + \beta^{-1} \partial_p \ln \rho_{\beta}) a^{\minus},
\label{e:hcl}
\end{equation}
where $ \hat{\cL}_0 $ is defined at the end of Proposition~\ref{prop:detailed},
so that
\begin{equation}
D =  \beta^{-1} \frac{dU}{dF} + \frac{1}{2 \gamma}
  \left\langle \left((\cL^{-1} + \hcL^{-1})\mathbb{P} p\right) 
\mathbb{P}  V^{\prime} (q) \right \rangle_{\rho}. \label{e:maesdetshort}
 \end{equation}
In both of these expressions, we note that $ \mathbb{P}p = p - U $, from Eq.~\eqref{e:drift_defn}.
 
Inspecting expression Eq.~\eqref{e:maesdetlong} for the correction to the Einstein relation, we see that beyond computing $ \rho_{\beta} $ as the stationary solution of the Fokker-Planck equation~\eqref{e:fp_station}, we must solve a Poisson equation of the form~\eqref{e:cell} as well as a second Poisson equation of the form
\begin{equation*}
- \cL \psi = \mathbb{P} V^{\prime}.
\end{equation*}
In the expression Eq.~\eqref{e:maesdetshort}, we must solve a stationary Fokker-Planck equation~\eqref{e:fp_station}, a Poisson equation of the form~\eqref{e:cell}, as well as an adjoint-Poisson equation of the form
\begin{equation*}
- \hcL \eta = \mathbb{P} p.
\end{equation*}
In both cases, it seems that an additional equation would need to be solved beyond the stationary Fokker-Planck equation~\eqref{e:fp_station}
 and a single Poisson equation~\eqref{e:cell} necessary in the homogenization approach.  On the other hand, computing the mobility $ \frac{dU}{dF} $ at general values of the tilt $ F $ from the nonperturbative homogenization equations would require a differentiation between different values of $ F$.  The direct formula~\eqref{e:maes} would generally of course need to be evaluated through Monte Carlo averages involving a large number of sample trajectories run for sufficiently long.
 
The perturbation theory with respect to $ F$ developed for the homogenization equations in Section~\ref{sec:lin_resp} has the virtue of allowing the simultaneous numerical computation of the diffusivity and drift for a range of values of tilt $ F$, rather than one value at a time.  One could introduce similar perturbation expansions with respect to tilt $ F$ into the formulas~\eqref{e:maesdetlong} and~\eqref{e:maesdetshort}.  We attempted to examine whether this would give equivalent results, but found this effort frustrating.  On the one hand, computing Eq.~\eqref{e:maesdetlong} perturbatively would introduce the perturbative series solution to a second Poisson equation completely absent from the homogenization theory, so it would be difficult to relate the results.  Expression~\eqref{e:maesdetshort} has more promise because to leading order, $ \hcL^{-1} $ is identical to the simple operator $ \hat{\cL}_0^{-1} $, which is just a time reversal of the operator $ \cL_0^{-1} $.  However, implementing the perturbation expansion on Eq.~\eqref{e:maesdetshort}, even to first order, produced considerably more unwieldy equations than emerged from the homogenization equations, and again how to relate the resulting expressions was unclear.  The main complication is the propagation of the perturbation expansion~\eqref{e:fexp} for the invariant density through the adjoint operator $ \hcL$~\eqref{e:hcl}.   Perhaps a more clever analysis would provide a linkage between the formula for the correction~\eqref{e:maes} to the Einstein relation from~\cite{BaiesiMaes11}  and the perturbative expansion we have developed in Proposition~\ref{prop:detailed}, but it appears that computations are considerably simpler by conducting the  perturbation expansion on the homogenization equations as we have done in Section~\ref{sec:lin_resp}.

\def\cprime{$'$} \def\cprime{$'$} \def\cprime{$'$} \def\cprime{$'$}
  \def\cprime{$'$} \def\cprime{$'$} \def\cprime{$'$}
  \def\Rom#1{\uppercase\expandafter{\romannumeral #1}}\def\u#1{{\accent"15
  #1}}\def\Rom#1{\uppercase\expandafter{\romannumeral #1}}\def\u#1{{\accent"15
  #1}}\def\cprime{$'$} \def\cprime{$'$} \def\cprime{$'$} \def\cprime{$'$}
  \def\cprime{$'$} \def\cprime{$'$} \def\cprime{$'$}
  \def\polhk#1{\setbox0=\hbox{#1}{\ooalign{\hidewidth
  \lower1.5ex\hbox{`}\hidewidth\crcr\unhbox0}}} \def\cprime{$'$}
  \def\cprime{$'$}

%\bibliography{../../../../bibtex_files/mybib}

\begin{thebibliography}{10}

\bibitem{Ferrando_al02}
Ying~SC. Ala-Nissila~T, Ferrando~R.
\newblock Collective and single particle diffusion on surfaces.
\newblock {\em Advances in Physics}, 51(3):949--1078, 2002.

\bibitem{BaiesiMaes11}
M.~Baiesi, C.~Maes, and B.~Wynants.
\newblock The modified {S}utherland-{E}instein relation for diffusive
  non-equilibria.
\newblock {\em Proc. R. Soc. Lond. Ser. A Math. Phys. Eng. Sci.},
  467(2134):2792--2809, 2011.

\bibitem{Caratti_all96}
G.~Caratti, R.~Ferrando, R.~Spadacini, and G.~E. Tommei.
\newblock Noise-activated diffusion in the egg-carton potential.
\newblock {\em Phys. Rev. E}, 54(5):4708--4721, Nov 1996.

\bibitem{ferrando_all92}
G.~Caratti, R.~Ferrando, R.~Spadacini, and G.E. Tommei.
\newblock An analytical approximation to the diffusion coefficient in
  overdamped multidimensional systems.
\newblock {\em Physica A}, 246:115--131, 1997.

\bibitem{coffey04}
W.T. Coffey, Y.P. Kalmykov, and J.T. Waldron.
\newblock {\em The {L}angevin equation}.
\newblock World Scientific, Singapore, 2004.

\bibitem{ColletMatinez2008}
P.~Collet and S.~Mart{\'{\i}}nez.
\newblock Asymptotic velocity of one dimensional diffusions with periodic
  drift.
\newblock {\em J. Math. Biol.}, 56(6):765--792, 2008.

\bibitem{CostantiniMarchesoni1999}
G~Costantini and F~Marchesoni.
\newblock {Threshold diffusion in a tilted washboard potential}.
\newblock {\em {Europhysics Letters}}, {48}({5}):{491--497}, {DEC} {1999}.

\bibitem{DoerDonKlos98}
C.R. Doering, L.~A. Dontcheva, and M.M. Klosek.
\newblock Constructive role of noise: fast fluctuation asymptotics of transport
  in stochastic ratchets.
\newblock {\em Chaos}, 8(3):643--649, 1998.

\bibitem{EvstigReimann2008}
M.~Evstigneev, O.~Zvyagolskaya, S.~Bleil, R.~Eichhorn, C.~Bechinger, and
  P.~Reimann.
\newblock Diffusion of colloidal particles in a tilted periodic potential:
  Theory versus experiment.
\newblock {\em Physical Review E}, 77(4):041107, April 2008.

\bibitem{FokGuoTang2002}
J.~C.~M. Fok, B.~Guo, and T.~Tang.
\newblock Combined {H}ermite spectral-finite difference method for the
  {F}okker-{P}lanck equation.
\newblock {\em Math. Comp.}, 71(240):1497--1528 (electronic), 2002.

\bibitem{GlimmJaffe87}
J.~Glimm and A.~Jaffe.
\newblock {\em Quantum physics}.
\newblock Springer-Verlag, New York, second edition, 1987.
\newblock A functional integral point of view.

\bibitem{HairPavl04}
M.~Hairer and G.A. Pavliotis.
\newblock Periodic homogenization for hypoelliptic diffusions.
\newblock {\em J. Statist. Phys.}, 117(1-2):261--279, 2004.

\bibitem{Heinsalu2004}
Els Heinsalu, Risto Tammelo, and Teet \"{O}rd.
\newblock Diffusion and current of brownian particles in tilted piecewise
  linear potentials: Amplification and coherence.
\newblock {\em Physical Review E}, 69(2):021111, February 2004.

\bibitem{HelNi05}
B.~Helffer and F.~Nier.
\newblock {\em Hypoelliptic estimates and spectral theory for {F}okker-{P}lanck
  operators and {W}itten {L}aplacians}, volume 1862 of {\em Lecture Notes in
  Mathematics}.
\newblock Springer-Verlag, Berlin, 2005.

\bibitem{HorsLef84}
W.~Horsthemke and R.~Lefever.
\newblock {\em Noise-induced transitions}, volume~15 of {\em Springer Series in
  Synergetics}.
\newblock Springer-Verlag, Berlin, 1984.
\newblock Theory and applications in physics, chemistry, and biology.

\bibitem{JoubaudStoltz2011}
R.~Joubaud and G.~Stoltz.
\newblock Nonequilibrium shear viscosity computations with langevin dynamics.
\newblock {\em arXiv preprint 1106.0633}, 2011.

\bibitem{KomOlla2005}
T.~Komorowski and S.~Olla.
\newblock On mobility and {E}instein relation for tracers in time-mixing random
  environments.
\newblock {\em J. Stat. Phys.}, 118(3-4):407--435, 2005.

\bibitem{KuboTodaHashitsume91}
R.~Kubo, M.~Toda, and N.~Hashitsume.
\newblock {\em Statistical physics. {II}}, volume~31 of {\em Springer Series in
  Solid-State Sciences}.
\newblock Springer-Verlag, Berlin, second edition, 1991.
\newblock Nonequilibrium statistical mechanics.

\bibitem{lebo_einstein}
J.L. Lebowitz and H.~Rost.
\newblock The {E}instein relation for the displacement of a test particle in a
  random environment.
\newblock {\em Stochastic Process. Appl.}, 54(2):183--196, 1994.

\bibitem{PavSt05b}
G.~A. Pavliotis and A.~M. Stuart.
\newblock Periodic homogenization for inertial particles.
\newblock {\em Phys. D}, 204(3-4):161--187, 2005.

\bibitem{PavlVog08}
G.~A. Pavliotis and A.~Vogiannou.
\newblock Diffusive transport in periodic potentials: Underdamped dynamics.
\newblock {\em Fluct. Noise Lett.}, 8(2):L155--173, 2008.

\bibitem{pavl05}
G.A. Pavliotis.
\newblock A multiscale approach to {B}rownian motors.
\newblock {\em Phys. Lett. A}, 344:331--345, 2005.

\bibitem{PavlSt08}
G.A. Pavliotis and A.M. Stuart.
\newblock {\em Multiscale methods}, volume~53 of {\em Texts in Applied
  Mathematics}.
\newblock Springer, New York, 2008.
\newblock Averaging and homogenization.

\bibitem{reimann_al02}
P.~Reimann, C.~Van den Broeck, H.~Linke, P.~H\"{a}nggi, J.M. Rubi, and
  A.~Perez-Madrid.
\newblock Diffusion in tilted periodic potentials: enhancement, universality
  and scaling.
\newblock {\em Phys. Rev. E}, 65(3):031104, 2002.

\bibitem{reimann_al01}
P.~Reimann, C.~Van den Broeck, H.~Linke, J.M. Rubi, and A.~Perez-Madrid.
\newblock Giant acceleration of free diffusion by use of tilted periodic
  potentials.
\newblock {\em Phys. Rev. Let.}, 87(1):010602, 2001.

\bibitem{ResibDeLeen77}
P.~Resibois and M.~De Leener.
\newblock {\em Classical Kinetic Theory of Fluids}.
\newblock Wiley, New York, 1977.

\bibitem{Ris84}
H.~Risken.
\newblock {\em The {F}okker-{P}lanck equation}, volume~18 of {\em Springer
  Series in Synergetics}.
\newblock Springer-Verlag, Berlin, 1989.

\bibitem{rodenh}
H.~Rodenhausen.
\newblock Einstein's relation between diffusion constant and mobility for a
  diffusion model.
\newblock {\em J. Statist. Phys.}, 55(5-6):1065--1088, 1989.

\bibitem{Straton58}
R.~L. Stratonovich.
\newblock Synchronization of an oscillator in the presence of interference.
\newblock {\em Radiotekh. Elektron. (Moscow)}, 3(4):497--506, 1958.

\bibitem{Straton67}
R.~L. Stratonovich.
\newblock {\em Topics in the theory of random noise. {V}ol. {II}}.
\newblock Revised English edition. Translated from the Russian by Richard A.
  Silverman. Gordon and Breach Science Publishers, New York, 1967.

\bibitem{Tit78}
U.~M. Titulaer.
\newblock A systematic solution procedure for the {F}okker-{P}lanck equation of
  a {B}rownian particle in the high-friction case.
\newblock {\em Phys. A}, 91(3-4):321--344, 1978.

\end{thebibliography}
\bibliographystyle{plain}
\end{document}